\documentclass[a4paper,UKenglish]{lipics-v2016}

\usepackage{amssymb}

\usepackage{colortbl , multirow}

\usepackage{wrapfig , graphicx}

\usepackage{mystyle_lipics}

\title{Hierarchical Time-Dependent Oracles\footnote{%
	This work is partially supported by the EU FP7/2007-2013, under grant agreements no.~609026 (project MOVESMART) and no.~621133 (project HoPE), and by DFG under grant agreement no. FOR-2083.
	}
}

\titlerunning{Hierarchical Time-Dependent Oracles} 

\author[1,4]{Spyros Kontogiannis}
\author[2]{Dorothea Wagner}
\author[3,4]{Christos Zaroliagis}

\affil[1]{%
	Comp. Sci. \& Eng. Dept.,
	University of Ioannina, 
	Greece
  \hfill[\texttt{kontog@cse.uoi.gr}]}

\affil[2]{%
	Karlsruhe Institute of Technology,
	Germany
	\hfill[\texttt{dorothea.wagner@kit.edu}]}

\affil[3]{%
	Comp. Eng. \& Inf. Dept.,
	University of Patras, 
	Greece
	\hfill[\texttt{zaro@ceid.upatras.gr}]}

\affil[4]{%
	Computer Technology Institute and Press ``Diophantus'',
	Greece
	}

\authorrunning{S.~Kontogiannis, \ D.~Wagner, \ C.~Zaroliagis}

\Copyright{Spyros Kontogiannis, Dorothea Wagner and Christos Zaroliagis}

\subjclass{%
	F.2.2 Nonnumerical Algorithms and Problems; 
	I.1.2 Algorithms 
	05C85 Graph algorithms; 
	05C12 Distance in graphs; 
}

\keywords{%
	Time-dependent shortest paths;
	FIFO property;
	Distance oracles.
}

\begin{document}

\maketitle

\begin{abstract}

We study networks obeying \emph{time-dependent} min-cost path metrics, and present novel oracles for them which \emph{provably} achieve two unique features: 
(i) \emph{subquadratic} preprocessing time and space, \emph{independent} of the metric's amount of disconcavity; 
(ii) \emph{sublinear} query time, in either the network size or the actual Dijkstra-Rank of the query at hand.

\end{abstract}

\section{Introduction}
\label{section:intro}

Concurrent technological infrastructures (e.g., road networks, social networks, e-commerce platforms, energy-management systems) are typically of very large scale and impose as a routine task the computation of min-cost paths in real-time, while their characteristics usually evolve with time.
The large-scale and real-time response challenges have been addressed in the last $15$ years by means of a new algorithmic trend: the provision of \emph{oracles}. That is, data structures created by appropriately selecting precomputed information (summaries) and which subsequently support query algorithms with real-time responses. The quality of an oracle is assessed by its preprocessing space and time requirements, the time-complexity of the query algorithm and the approximation guarantee (stretch).
Numerous oracles have been proposed and analyzed (see e.g., \cite{2013-Agarwal-Godfrey,2010-Patrascu-Roditty,2011-Porat-Roditty,2014-Sommer-spq-survey,2009-Sommer-Verbin-Yu,2005-Thorup-Zwick,2012-Wulf-Nilsen-a,2012-Wulf-Nilsen-b} and references therein) for large-scale, mostly undirected networks, accompanied by a \emph{static} arc-cost metric.
In tandem with oracles, an equally important effort (with similar characteristics) has also emerged in the last $15$ years under the tag of \term{speedup techniques}, for approaches tailored to work extremely well in real-life instances (see e.g., \cite{2014-Bast-Delling-Goldberg-Hannemann-Pajor-Sanders-Wagner-Werneck} and references therein).

The temporality of the network characteristics is often depicted by some kind of predetermined dependence of the metric on the actual time that each resource is used (e.g., traversal speed in road networks, packet-loss rate in IT networks, arc availability in social networks, etc). Perhaps the most typical application scenario (motivating also our work) is route planning in road networks, where the travel-time for traversing an arc $a=uv$ (modeling a road segment) depends on the temporal traffic conditions while traversing $uv$, and thus on the departure time from its tail $u$.
This gives rise to \emph{time-varying} network models and to computing min-cost (a.k.a. shortest) paths in such networks. Several variants of this problem try to model time-variation of the underlying graph structure and/or the arc-cost metric (e.g., dynamic shortest paths, parametric shortest paths, stochastic shortest paths, temporal networks, etc).
In this work, we consider the case in which the cost variation of an arc $a$ is determined by a \emph{function} $D[a]$, which is considered to be a continuous, piecewise linear (pwl)\footnote{
	Major car navigator vendors provide real-time estimations of travel-time values by periodically sampling the average speed of road segments using the cars connected to the service as sampling devices. The most customary way to represent this historic traffic data, is to consider the continuous pwl interpolants of the sample points as \emph{arc-travel-time functions} of the corresponding instance.} 
	and periodic function of the time at which the resource is actually being used \cite{2004-Dean-a,2012-Dehne-Omran-Sack,2014-Foschini-Hershberger-Suri,1990-Orda-Rom}.
In case of providing route plans in time-dependent road networks, arc-costs are considered as \term{arc-travel-time} functions, while time-dependent shortest paths as \term{minimum-travel-time} paths.
The goal is then to determine the cost (\term{minimum-travel-time}) of a shortest path from an origin $o$ to a destination $d$, as a function of the \term{departure-time} $t_o$ from $o$. Due to the time-dependence of the arc-cost metric, the actual arc-cost value of an arc $a =  uv$ is unknown until the exact time $t_u \geq t_o$ at which $uv$ starts being traversed.

\smallskip
\noindent
\emph{\textbf{Problem Setting and Related Work.}}
Two variants of the \term{time-dependent shortest path} problem have been considered in the literature:
$TDSP(o,d,t_o)$ (resp.~$TDSP(o,\star,t_o)$) focuses on the one-to-one (resp.~one-to-all) determination of the \emph{scalar cost} of a minimum-travel-time (shortest) path to $d$ (resp.~for all $d$), when departing from the given origin $o$ at time $t_o$.
$TDSP(o,d)$ (resp.~$TDSP(o,\star)$) focuses on the one-to-one (resp., one-to-all) succinct representation of the time-dependent minimum-travel-time path \emph{function}(s) $D[o,d]$ from  $o$ to $d$ (resp.~towards all reachable $d$), and all departure-times from $o$.


$TDSP(o,d,t_o)$ has been studied as early as \cite{1966-Cooke-Halsey}. The first work on $TDSP(o,d,t_o)$ for continuous time-axis was \cite{1969-Dreyfus} where it was proved that, if \emph{waiting-at-nodes} is allowed unconditionally, then $TDSP(o,d,t_o)$ is solvable in quasilinear time via a time-dependent variant of Dijkstra's algorithm (we call it $\alg{TDD}$), which relaxes arcs by computing the arc costs ``on the fly'', upon settling their tails.
A more complete treatment for the continuous case, considering various limitations in the waiting-times at the nodes of the network, was provided in \cite{1977-Halpern}; an algorithm was also given for $TDSP(o,d,t_o)$, whose complexity cannot be bounded by a function of the network topology.
An excellent overview of the problem is provided in \cite{1990-Orda-Rom}. Among other results, it was proved that for \emph{affine arc-cost} functions possessing the FIFO property (according to which all the arc-cost functions have slopes at least $-1$), in addition to $\alg{TDD}$, a time-dependent variant of the label-correcting Bellman-Ford algorithm works.
Moreover, if waiting-at-nodes is \emph{forbidden} and the arc-costs do not preserve the FIFO property, then \term{subpath-optimality} of shortest paths is not necessarily preserved. In such a case, many variants of the problem are also NP-hard \cite{SOS1998}. Additionally, when shortest path costs are well defined and optimal waiting-times at nodes always exist, a non-FIFO arc with \emph{unrestricted-waiting-at-tail} policy is equivalent to a FIFO arc in which waiting at the tail is not beneficial \cite{1990-Orda-Rom}.
For these reasons, we focus here on instances for which the FIFO property holds, as indeed is the case with most of past and recent work on $TDSP(o,d,t_o)$.


The complexity of $TDSP(o,d)$ was first questioned in \cite{2004-Dean-b,2004-Dean-a} and remained open until recently, when it was proved in \cite{2014-Foschini-Hershberger-Suri} that, in case of FIFO-abiding pwl arc-cost functions, for a \emph{single} origin-destination pair $(o,d)$ the space complexity for succinctly representing $D[o,d]$ is $(1+K)\cdot n^{\Theta(\log n)}$, where $n$ is the number of vertices and $K$ is the total number of breakpoints (or legs) of all the arc-cost functions. Note that $K$ can be substituted by the number $K^*$ of \emph{concavity-spoiling} breakpoints (at which the arc-cost slopes increase) of the arc-cost functions. Several \emph{output-sensitive} algorithms for the exact computation of $D[o,d]$ have been presented \cite{2004-Dean-a,2012-Dehne-Omran-Sack,2014-Foschini-Hershberger-Suri,1990-Orda-Rom}, the most efficient being the ones in \cite{2012-Dehne-Omran-Sack,2014-Foschini-Hershberger-Suri}.


Due to the above mentioned hardness of $TDSP(o,d)$, and also since the time-dependent arc-costs are typically only (e.g., pwl) approximations of the actual costs, it is quite natural to seek for succinct representations of approximations to $D[o,d]$, which aim at trading-off accuracy for computational effort. Several \emph{one-to-one} $(1+\eps)$-approximation algorithms for $TDSP(o,d)$ appeared recently in the literature~\cite{2012-Dehne-Omran-Sack,2014-Foschini-Hershberger-Suri,2014-Omran-Sack}, the most successful being those provided in~\cite{2014-Omran-Sack}.
\remove{
A one-to-one (i.e., for a single $(o,d)$ pair) approximation algorithm for computing a $D[o,d]$ function was provided in \cite{2012-Dehne-Omran-Sack}. That algorithm requires
\(
	\Order{\frac{1}{\eps}\cdot(D_{\max}[o,d] - D_{\min}[o,d])}
\)
calls of $TDSP(d,o,t_d)$ in the reverse instance for producing a $(1+\eps)$-upper-approximating function $\De[o,d]$ of $D[o,d]$, where $D_{\max}[o,d] = \max_{t\in[0,T)} D[o,d](t)$, $D_{\min}[o,d] = \min_{t\in[0,T)} D[o,d](t)$, and $T$ is the time period.
Another one-to-one $(1+\eps)-$approximation algorithm was provided in \cite{2014-Foschini-Hershberger-Suri}. That algorithm makes
\(
	\Order{\frac{K}{\eps}\log\left(\frac{D_{\max}[o,d]}{D_{\min}[o,d]}\right) \cdot\log\left(\frac{T}{K\eps D_{\min}[o,d]}\right)}
\)
calls to $TDSP(o,d,t_o)$.

Two further one-to-one $(1+\eps)$-approximation algorithms for $D[o,d]$ were given in \cite{2014-Omran-Sack}. The first algorithm requires
\(
	\Order{\frac{K}{\eps}\cdot
	\left[
			\log\left(\frac{D_{\max}[o,d]}{D_{\min}[o,d]}\right) + \log\left(\frac{T}{K D_{\min}[o,d]}\right)
	\right]}
\)
and the second algorithm requires
\(
	\Order{ K \cdot
		\left[ \frac{1}{\eps}
			\log\left(\frac{D_{\max}[o,d]}{D_{\min}[o,d]}\right) + \log\left(\frac{L}{K\eps D_{\min}[o,d]}\right)\right]}
\)
calls to $TDSP(o,d,t_o)$.
}
The first \emph{one-to-all} $(1+\eps)$-approximation algorithm for $TDSP(o,\star)$, called \term{bisection} ($\alg{BIS}$), was given in \cite{2014-Kontogiannis-Zaroliagis}. It is based on bisecting the (common to all functions) axis of departure-times from $o$ and considers slightly stricter assumptions than just the FIFO property for the arc-cost metric.
$\alg{BIS}$ requires $\Order{\frac{K^*}{\eps}\cdot\log^2\left(\frac{n}{\eps}\right)}$ calls to $TDSP(o,\star,t_o)$. Note that all one-to-one approximation algorithms for $TDSP(o,d)$ \cite{2012-Dehne-Omran-Sack,2014-Foschini-Hershberger-Suri,2014-Omran-Sack} demand, in worst-case, a comparable amount of calls to $TDSP(o,\star,t_o)$, just for one $od$-pair.


Minimum-travel-time oracles for time-dependent networks (TD-oracles henceforth) had received no attention until recently \cite{2014-Kontogiannis-Zaroliagis}. A TD-oracle is an offline-precomputed data structure that allows the evaluation of an \emph{upper-approximation} $\overline{\Delta}[o,d](t_o)$ of $D[o,d](t_o)$ (due to the hardness of computing and storing $D[o,d]$), for \emph{any} possible query $(o,d,t_o)\in V\times V\times\nonnegativereals$ that may appear in an online fashion.
One trivial solution would be to provide a succinct representation of $\overline{\Delta}[o,d]$ for each pair $(o,d)\in V\times V$, for the sake of rapid evaluations in the future but at the expense of superquadratic space. Another trivial solution would be to execute $\alg{TDD}$ ``on-the-fly'' for each new query $(o,d,t_o)$, at the expense of superlinear query-time.
A non-trivial TD-oracle should aim to trade-off smoothly preprocessing requirements with query-response times and approximation guarantees. In particular, it should precompute a data structure in \emph{subquadratic} time and space, and also provide a query algorithm which evaluates \emph{efficiently} (i.e., faster than $\alg{TDD}$) $\overline{\Delta}[o,d](t_o)$, where $\overline{\Delta}[o,d]$ must have a \emph{provably} good approximation guarantee w.r.t. $D[o,d]$.
Note that there exists important applied work (\term{speedup heuristics}) for computing time-dependent shortest paths (e.g., \cite{2013-Batz-Geisberger-Sanders-Vetter,2011-Delling_TDSHARC,2009-Delling-Wagner,ndls-bastd-12}), which however provide mainly empirical evidence on the success of the adopted approaches.

The TD-oracles in \cite{2014-Kontogiannis-Zaroliagis} required $\Order{n^{2-\beta} (K^*+1)}$ preprocessing space and time, for some constant $\beta\in(0,1)$, and are able to answer queries (under certain conditions) in time $\Order{n^{\delta}}$, for some constant $\delta\in(0,1)$. When $K^*\in \order{n}$, the oracles can be fine-tuned so as to assure sublinear query-times and subquadratic preprocessing requirements in $n$.
An extensive experimental evaluation of those oracles on a real-world road network is provided in \cite{2015-Kontogiannis-Michalopoulos-Papastavrou-Paraskevopoulos-Wagner-Zaroliagis}, demonstrating their practicality, at the expense, however, of large memory consumption due to the linear dependence of the preprocessing space on $K^*$ which can be $\Omega(n)$.

The main challenge addressed here is to provide TD-oracles that achieve:
(i) subquadratic preprocessing requirements, \emph{independently} of $K^*$; and
(ii) query-times sublinear, not only in the worst-case (i.e., in $n$), but also in the number $\Gamma[o,d](t_o)$ of settled vertices when executing $\alg{TDD}(o,\star,t_o)$ until $d$ is settled (\emph{Dijkstra-Rank}).

\smallskip
\noindent
\emph{\textbf{Our Contributions.}}
We address positively the aforementioned challenge by providing:
	\emph{(i)}~A novel and remarkably simple algorithm ($\alg{TRAP}$) (cf. Section~\ref{section:TRAP-method}) for constructing one-to-many $(1+\eps)$-upper-approximations $\overline{\Delta}[o,d]$ (\emph{summaries}) of minimum-travel-time functions $D[o,d]$, for all ``sufficiently distant'' destinations $d$ from the origin $o$. $\alg{TRAP}$ requires $\order{n}$ calls to $TDSP(o,\star,t_o)$, which is \emph{independent} of the degree of concavity $K^*$. Its novelty is that it does not demand the concavity of the unknown function to approximate.
	\emph{(ii)}~The $\alg{TRAPONLY}$ and $\alg{FLAT}$ oracles (cf. Section~\ref{section:FLAT-oracle}) which exploit $\alg{TRAP}$ and $\alg{BIS}$ to construct minimum-travel-time summaries from randomly selected landmarks to \emph{all} reachable destinations. The preprocessed data structures require subquadratic space and time, independently of $K^*$. $\alg{FLAT}$ uses the query algorithms in \cite{2014-Kontogiannis-Zaroliagis}. $\alg{TRAPONLY}$ needs to extend them in order to recover missing summaries for local neighborhoods around a landmark. In both cases sublinear query-times are achieved.
	\emph{(iii)}~The $\alg{HORN}$ oracle (cf. Section~\ref{section:HORN-oracle}) which organizes a hierarchy of landmarks, from many local landmarks possessing summaries only for small neighborhoods of destinations around them, up to a few global landmarks possessing summaries for all reachable destinations. $\alg{HORN}$'s preprocessing requirements are again subquadratic. We then devise and analyze a novel query algorithm ($\alg{HQA}$) which exploits this hierarchy, with query-time \emph{sublinear} in the Dijkstra-Rank of the query at hand.

Except for the choice of landmarks, our algorithms are deterministic. A recent experimental study \cite{2016-Kontogiannis-Michalopoulos-Papastavrou-Paraskevopoulos-Wagner-Zaroliagis} demonstrates the excellent performance of our oracles in practice, achieving considerable memory savings and query times about three orders of magnitude faster than $\alg{TDD}$, and more than $70\%$ faster than those in \cite{2015-Kontogiannis-Michalopoulos-Papastavrou-Paraskevopoulos-Wagner-Zaroliagis}.
Table~\ref{table:td-oracles-overview} summarizes the achievements of the TD-oracles presented here and their comparison with the oracles in \cite{2014-Kontogiannis-Zaroliagis}.
\begin{table}[t]
\centering
\begin{tabular}{|r||c|c|c|}
\hline
& preprocessing space/time & query time & recursion budget (depth) $r$
\\ \hline\hline
\cite{2014-Kontogiannis-Zaroliagis} 
& $K^*\cdot n^{2-\beta+\order{1}}$ & $n^{\delta+\order{1}}$ & $r\in\Order{1}$
\\ \hline
$\alg{TRAPONLY}$
& $n^{2-\beta+\order{1}}$ & $n^{\delta+\order{1}}$ & $r\approx\frac{\delta}{a} - 1$
\\ \hline
$\alg{FLAT}$
& $n^{2-\beta+\order{1}}$ & $n^{\delta+\order{1}}$ & $r\approx\frac{2\delta}{a} - 1$
\\ \hline
$\alg{HORN}$
& $n^{2-\beta+\order{1}}$ & $\approx \Gamma[o,d](t_o)^{\delta+\order{1}}$ & $r\approx\frac{2\delta}{a} - 1$
\\ \hline
\end{tabular}
\caption{\label{table:td-oracles-overview} Achievements of oracles for TD-instances with period $T=n^a$, for constant $a\in(0,1)$. The stretch of all query algorithms is $1+\eps\cdot \frac{ (\eps / \psi)^{r+1} }{ (\eps / \psi)^{r+1} - 1 }$. For all oracles, except for the first, we assume that $\beta\downarrow 0$.}
\end{table}

\section{Preliminaries}
\label{section:preliminaries+assumptions}

\smallskip
\noindent
\emph{\textbf{Notation and Terminology.}} 
For any integer $k\geq 1$, let $[k] = \{ 1,2,\ldots, k\}$.
A \term{time-dependent network instance} (TD-instance henceforth) consists of a directed graph $G=(V,A)$ with $|V| = n$ vertices and $|A|= m \in \Order{n}$ arcs, where each arc $a\in A$ is accompanied with a continuous, pwl \emph{arc-cost} function $D[a]:\nonnegativereals\mapsto \positivereals$. We assume that all these functions are periodic with period $T>0$ and are defined as follows: $\forall k\in \naturals, \forall t\in [0,T),~ D[a](kT+t) = d[a](t)$, where $d[a]: [0,T) \rightarrow (0,M_a]$ is such that $\lim_{t\uparrow T}d[a](t) = d[a](0)$, for some fixed integer $M_a$ denoting the maximum possible cost ever seen for arc $a$. Let also $M = \max_{a\in A} M_a$ denote the maximum arc-cost ever seen in the entire network.
Since $D[a]$ is periodic, continuous and pwl function, it can be represented succinctly by a sequence of $K_{a}$ breakpoints (i.e., pairs of departure-times and arc-cost values) defining $d[a]$. $K = \sum_{a\in A} K_{a}$ is the number of breakpoints representing all  arc-cost functions, $K_{\max} = \max_{a\in A} K_{a}$, and $K^*$ is the number of \emph{concavity-spoiling} breakpoints (the ones at which the arc-cost function slopes increase). Clearly, $K^* \leq K$, and $K^* = 0$ for \emph{concave} arc-cost functions.

To ease the exposition and also for the sake of compliance with terminology in previous works (inspired by the primary application scenario of route planning in time-dependent road networks), we consider arc-costs as \emph{arc-travel-times} and time-dependent shortest paths as \term{minimum-travel-time} paths. This terminology facilitates the following definitions.
The \term{arc-arrival-time} function of $a\in A$ is $Arr[a](t) = t + D[a](t)$, $\forall t\in [0,\infty)$.
The \emph{path-arrival-time} function of a path $p = \langle a_1,\ldots,a_k \rangle$ in $G$ (represented as a sequence of arcs) is the composition $Arr[p](t) = Arr[a_k](Arr[a_{k-1}](\cdots(Arr[a_1](t))\cdots))$ of the arc-arrival-time functions for the constituent arcs.
The \emph{path-travel-time} function is then $D[p](t) = Arr[p](t) - t$.

For any $(o,d)\in V\times V$, $\mathcal{P}_{o,d}$ denotes the set of $od$-paths. For any $p\in \mathcal{P}_{o,x}$ and $q\in\mathcal{P}_{x,d}$, $s = p\bullet q \in \mathcal{P}_{o,d}$ is the concatenation of $p$ and $q$ at $x$.
The \term{earliest-arrival-time} function is 
	$Arr[o,d](t_o) 	= \min_{p\in \mathcal{P}_{o,d}}\left\{ Arr[p](t_o)\right\}$, $\forall t_o\geq 0$, 
while the \term{minimum-travel-time} function is defined as
	$D[o,d](t_o) 	= \min_{p\in \mathcal{P}_{o,d}}\left\{ D[p](t_o)\right\} = Arr[o,d](t_o) - t_o$. 
For a given query $(o,d,t_o)$, $SP[o,d](t_o) = \{ p\in P_{o,d} : Arr[p](t_o) = Arr[o,d](t_o) \}$ is the set of earliest-arrival-time (equivalently, minimum-travel-time) paths. 
$ASP[o,d](t_o)$ is the set of $od$-paths whose travel-time values are $(1+\eps)$-approximations of the minimum-travel-time among all $od$-paths.

When we say that we \emph{``grow a $\alg{TDD}$ ball from $(o,t_o)$''}, we refer to the execution of $\alg{TDD}$ from $o\in V$ at departure time $t_o\in [0,T)$ for solving $TDSP(o,\star,t_o)$ (resp.~$TDSP(o,d,t_o)$, for a specific destination $d$). Such a call, which we denote as $\alg{TDD}(o,\star,t_o)$ (resp.~$\alg{TDD}(o,d,t_o)$), takes time $\Order{m+n\log(n) [1+\log\log(1+K_{\max})]}$ $= \Order{n\log(n)\log\log(K_{\max})]}$, using predecessor search for evaluating continuous pwl functions (cf.~\cite{2014-Kontogiannis-Zaroliagis}).
The \term{Dijkstra-Rank} $\Gamma[o,d](t_o)$ is the number of settled vertices up to $d$, when executing $\alg{TDD}(o,d,t_o)$.

$\forall a=uv\in A$ and $[t_s,t_f)\subseteq [0,T)$, we define upper- and lower-bounding travel-time metrics:
	the \term{minimally-congested} travel-time $\underline{D}[uv](t_s,t_f) := \min_{t_u\in [t_s,t_f)} \{ D[uv](t_u) \}$  and
	the \term{maximally-congested} travel-time $\overline{D}[uv](t_s,t_f) := \max_{t_u\in [t_s,t_f)} \{ D[uv](t_u) \}$.
If $[t_s,t_f) = [0,T)$, we refer to the static \term{free-flow} and \term{full-congestion} metrics $\underline{D} , \overline{D} : A \rightarrow [1,M]$, respectively. Each arc $a\in A$ is also equipped with scalars $\underline{D}[a]$ and $\overline{D}[a]$ in these static metrics. For any arc-cost metric $D$, $diam(G,D)$ is the diameter (largest possible vertex-to-vertex distance) of the graph. For example, $diam(G,\underline{D})$ is the free-flow diameter of $G$.

In our TD-instance, we can guarantee that $T\geq diam(G,\underline{D})$.
	If this is not the case, we can take the minimum number $c$ of consecutive copies of each $d[a]$ as a single arc-travel-time function $d'[a]:[0,cT)\mapsto\positivereals$ and $D'[a](t+k T') = d'[a](t),~ \forall t\in [0, T')$ such that $T' = cT\geq diam(G,\underline{D'})$. 
In addition, we can also guarantee that $T = n^{\alpha}$ for a small \emph{constant} $\alpha\in(0,1)$ of our control.
	If $T \neq n^{\alpha}$, we scale the travel-time metric by setting $D'' = \frac{n^{\alpha}}{T }\cdot D$ (e.g., we change the unit by which we measure time, from milliseconds to seconds, or even minutes) and use the period $T'' = n^{\alpha}$, without affecting the structure of the instance at all.
From now on we consider w.l.o.g. TD-instances with $T = n^{\alpha} \geq diam(G,\underline{D})$.

For any $v\in V$, departure-time $t_v\in\nonnegativereals$, integer $F\in[n]$ and $R>0$, $B[v;F](t_v)$ ($B[v;R](t_v)$) is a ball of size $F$ (of radius $R$) grown by $\alg{TDD}$ from $(v,t_v)$, in the time-dependent metric. Analogously, $\underline{B}[v;F]$ ($\underline{B}[v;R]$) and $\overline{B}[v;F]$ ($\overline{B}[v;R]$) are, respectively, the size-$F$ (radius-$R$) balls from $v$ in the free-flow and fully-congested travel-time metrics.

A pair of continuous, pwl, periodic functions $\overline{\De}[o,d]$ and $\underline{\De}[o,d]$), with a (hopefully) small number of breakpoints, are \term{$(1+\eps)$-upper-approximation} and \term{$(1+\eps)$-lower-approximation} of $D[o,d]$, if 
	\(
		\forall t_o\geq 0,~
		\frac{D[o,d](t_o)}{1+\eps}
		\leq \underline{\De}[o,d](t_o)
		\leq D[o,d](t_o)
		\leq \overline{\De}[o,d](t_o)
		\leq (1+\eps)\cdot D[o,d](t_o)\,.
	\)


\smallskip
\noindent
\emph{\textbf{Assumptions on the time-dependent arc-cost metric.}} 
The directedness and time-dependence in the underlying network imply an asymmetric arc-cost metric that also evolves with time. To achieve a smooth transition from static and undirected graphs towards time-dependent and directed graphs, we need a quantification of the degrees of asymmetry and evolution of our metric over time. These are captured via a set of parameters depicting 
the steepness of the minimum-travel-time functions, 
the ratio of minimum-travel-times in opposite directions, 
and the relation between graph expansion and travel-times.
We make some assumptions on the values of these parameters which seem quite natural for our main application scenario (route planning in road networks), and were verified by an experimental analysis (cf. Appendix~\ref{section:assumptions}). Here we only present a qualitative interpretation of them. It is noted that Assumptions~\ref{assumption:Bounded-Travel-Time-Slopes} and \ref{assumption:Bounded-Opposite-Trips} were exploited also in the analyses in \cite{2014-Kontogiannis-Zaroliagis}.
%
\begin{assumption}[Bounded Travel-Time Slopes]
\label{assumption:Bounded-Travel-Time-Slopes}
	All the minimum-travel-time slopes are bounded in a given interval $[-\La_{\min},\La_{\max}]$, for given constants $\La_{\min}\in [0,1)$ and $\La_{\max}\geq 0$.
\end{assumption}
%
\begin{assumption}[Bounded Opposite Trips]
\label{assumption:Bounded-Opposite-Trips}
	The ratio of minimum-travel-times in opposite directions between two vertices, for any specific departure-time but not necessarily via the same path, is upper bounded by a given constant $\zeta\geq 1$.
	\end{assumption}
%
\begin{assumption}[Growth of Free-Flow Dijkstra Balls]
	\label{assumption:growth-of-free-flow-ball-sizes}
	$\forall F\in [n]$, the free-flow ball $\underline{B}[v;F]$ blows-up by at most a polylogarithmic factor, when expanding its (free-flow) radius up to the value of the full-congestion radius within $\underline{B}[v;F]$.
\end{assumption}
%

Finally, we need to quantify the correlation between the arc-cost metric and the Dijkstra-Rank metric induced by it. For this reason, inspired by the notion of the doubling dimension (e.g., \cite{2011-Bartal-Gottlieb-Kopelowitz-Lewenstein-Roditty} and references therein), we consider some \emph{scalar} $\lambda \geq 1$ and functions $f,g:\naturals\mapsto [1,\infty)$, such that the following hold:
$\forall (o,d,t_o)\in V\times V\times [0,T)$,
	(i) $\Gamma[o,d](t_o) \leq f(n) \cdot (D[o,d](t_o))^{\la}$, and
	(ii) $D[o,d](t_o) \leq g(n) \cdot (\Gamma[o,d](t_o))^{1/\la}$.
This property trivially holds, e.g., for $\la=1$, $f(n) = n$, and $g(n) = \max_{a\in A}\left\{ \overline{D}[a]\right\}$. Of course, our interest is for the smallest possible values of $\la$ and at the same time the slowest-growing functions $f(n),g(n)$.
Our last  assumption quantifies the boundedness of this correlation by restricting $\la$, $f(n)$ and $g(n)$.
\begin{assumption}
\label{assumption:Travel-Time-vs-Dijkstra-Rank}
		There exist $\la\in\order{\frac{\log(n)}{\log\log(n)}}$ and $f(n),g(n) \in \polylog(n)$ s.t. the following hold:
		(i) $\Gamma[o,d](t_o) \leq f(n) \cdot (D[o,d](t_o))^{\la}$, and
		(ii) $D[o,d](t_o) \leq g(n) \cdot (\Gamma[o,d](t_o))^{1/\la}$.
		Analogous inequalities hold for the free-flow and the full-congestion metrics $\underline{D}$ and $\overline{D}$.
\end{assumption}

Note that static oracles based on the doubling dimension (e.g., \cite{2011-Bartal-Gottlieb-Kopelowitz-Lewenstein-Roditty}) demand a \emph{constant} value for the exponent $\la$ of the expansion. We relax this by allowing $\la$ being expressed as a (sufficiently slowly) growing function of $n$. We also introduce some additional slackness, by allowing some divergence from the corresponding powers by polylogarithmic factors.

In the rest of the paper we consider sparse ($m \in \Order{n}$) TD-instances, compliant with Assumptions \ref{assumption:Bounded-Travel-Time-Slopes}, \ref{assumption:Bounded-Opposite-Trips}, \ref{assumption:growth-of-free-flow-ball-sizes}, and \ref{assumption:Travel-Time-vs-Dijkstra-Rank}. For convenience, the notation used throughout the paper is summarized in Appendix~\ref{section:notation}. A review of the oracles in \cite{2014-Kontogiannis-Zaroliagis} is presented in Appendix~\ref{section:2014-Kontogiannis-Zaroliagis-Overview}.

\section{The $\alg{TRAP}$ approximation method}
\label{section:TRAP-method}

\begin{wrapfigure}{r}{5.5cm}
\centerline{\includegraphics[width=5.5cm]{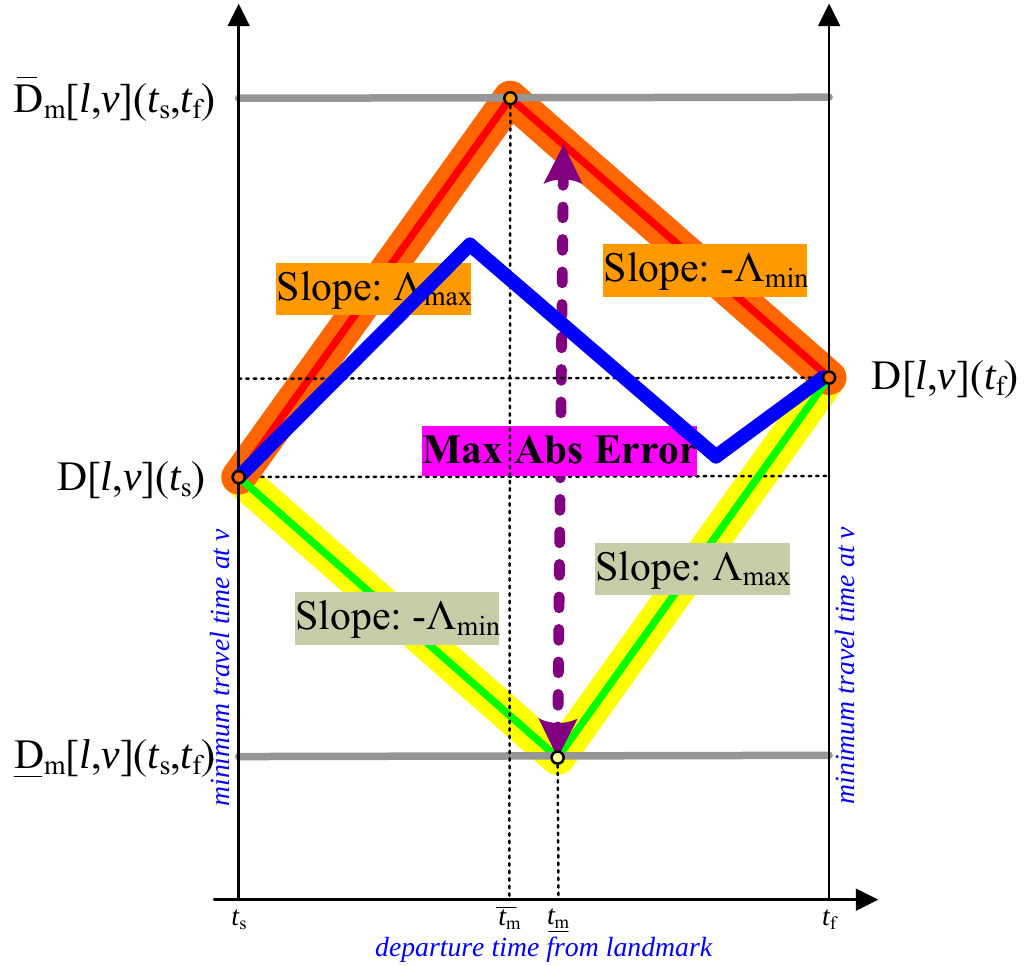}}
\caption{\label{fig:trapezoidal-approximation}
	\footnotesize{The upper-approximation $\overline{\delta}_k[\ell,v]$ (thick orange, upper pwl line), and lower-approximation $\underline{\delta}_k[\ell,v]$ (thick green, lower pwl line), of the unknown function $D[\ell,v]$ (blue pwl line) within the interval $I_k = [t_s=(k-1)\tau , t_f=k\tau)$.
}}
\vspace*{-20pt}
\end{wrapfigure}

We now introduce the \term{trapezoidal} ($\alg{TRAP}$) algorithm, a novel algorithm for computing one-to-many $(1+\eps)$-upper-approximations $\overline{\De}[\ell,v]:[0, T) \mapsto\positivereals$ of min-cost functions $D[\ell,v]$, from a vertex $\ell$ towards all sufficiently distant destinations (typically, $\ell$ will be a landmark).
	$\alg{TRAP}$ is remarkably simple and works as follows. First, $[0,T)$ is split into a number of $\ceil{\frac{T}{\tau}}$ consecutive length-$\tau$ subintervals, where $\tau$ is a tuning parameter to be fixed later. Then, for each such interval $[t_s, t_f = t_s + \tau) \subseteq [0,T)$, a $(1+\eps)$-upper-approximation of the projection $D[\ell,v]:[t_s,t_f) \mapsto\positivereals$ is computed. Finally, the concatenation of all these $(1+\eps)$-upper-approximations per subinterval constitutes the requested $(1+\eps)$-upper-approximation $\overline{\Delta}[\ell,v]$ of $D[o,d]:[0,T)\mapsto\positivereals$.
Note that, contrary to the $\alg{BIS}$ approximation algorithm \cite{2014-Kontogiannis-Zaroliagis}, no assumption is made on the shapes of the min-cost functions to approximate within each subinterval; in particular, no assumption is made on them being concave.
	$\alg{TRAP}$ only exploits the fact that $\tau$ is small, along with Assumption~\ref{assumption:Bounded-Travel-Time-Slopes} on the boundedness of travel-time slopes.	

We now describe the upper- and lower-approximations of $D[o,d]$ that we construct in a  subinterval $I_k = [t_s=(k-1)\tau, t_f = k\tau)\subset [0,T)$, $k\in\left[\ceil{\frac{T}{\tau}}\right]$, from a vertex $\ell\in V$ towards some destination $v\in V$. The quality of the upper-approximation depends on the value of $\tau$ and the delay values at the endpoints of $I_k$, as we shall explain shortly. $\alg{TRAP}$ computes the following two functions of $D[\ell,v]$ (cf.~Fig.~\ref{fig:trapezoidal-approximation}): $\forall t\in I_k$,
%
	\(
	\overline{\delta}_k[\ell,v](t) =
	\min	\left\{\begin{array}{c}
						D[\ell,v](t_f) + \La_{\min} t_f - \La_{\min} t~,~
						D[\ell,v](t_s) - \La_{\max} t_s + \La_{\max} t
		\end{array}\right\}
	\)
	and 
	\(
	\underline{\delta}_k[\ell,v](t) =
	\max	\left\{\begin{array}{c}
						D[\ell,v](t_f) - \La_{\max} t_f + \La_{\max} t ~,~
						D[\ell,v](t_s) + \La_{\min} t_s - \La_{\min} t
		\end{array}\right\}
	\)
%
and considers them as the upper- and lower-approximating functions of $D[\ell,v]$ within $I_k$. The correctness of this choice is proved in the next lemma, which follows by Assumption~\ref{assumption:Bounded-Travel-Time-Slopes}.
\begin{lemma}
\label{lemma:upper+lower-approximating-functions-of-TRAP}
	$\overline{\delta}_k[\ell,v](t)$ and $\underline{\delta}_k[\ell,v](t)$ upper- and lower-approximate $D[\ell,v](t)$ within $I_k$.
\end{lemma}
Let $(\underline{t}_m,\underline{D}_m)$ and $(\overline{t}_m,\overline{D}_m)$ be the intersections of the legs in the definition of $\underline{\delta}_k[\ell,v]$ and $\overline{\delta}_k[\ell,v]$, respectively.
The \term{maximum additive error} $MAE(I_k)$ for $\overline{\delta}_k[\ell,v]$
in $I_k$ (i.e., the length of the purple dashed line in Fig.~\ref{fig:trapezoidal-approximation}) is 
	\(
	MAE(I_k) := \max_{t\in I_k}\left\{ \overline{\delta}_k[\ell,v](t) - \underline{\delta}_k[\ell,v](t) \right\}
	=		\overline{\delta}_k[\ell,v](\underline{t}_m) - \underline{\delta}_k[\ell,v](\underline{t}_m)\,.
	\)
The following lemma proves that, for $\tau$ sufficiently small, $MAE(I_k)$ cannot be large. It also provides a \emph{sufficient condition} for the value of $\tau$ so that $\overline{\delta}_k[\ell,v]$ is indeed a $(1+\eps)$-upper-approximation of $D[\ell,v]$ within $I_k$.

\begin{lemma}
\label{lemma:hole-radius-of-trapezoidal}
$\forall(\ell,v)\in L\times V$, $\forall k\in\left[ \ceil{\frac{T}{\tau}} \right]$ and $I_k = [(k-1)\tau , k\tau)$, the following hold:
	(1) $MAE[\ell,v](I_k) \leq \La_{\max}\cdot\tau$;
	(2) $\overline{\delta}_k[\ell,v]$ is a $(1+\eps)$-upper-approximation of $D[\ell,v]$ within $I_k$, if 
		\(
			\left[~
				D[\ell,v](t_s) \geq \left(\La_{\min} +\frac{\La_{\max}}{\eps}\right)\cdot\tau
			~\right]
		\)
		$\OR$
		\(
			\left[~
		D[\ell,v](t_f) \geq \left(1 +\frac{1}{\eps}\right)\La_{\max}\cdot\tau
			~\right]
		\)
\end{lemma}
For given $\tau>0$ and $\ell\in L$, the set of \emph{faraway destinations} from $\ell$ is 
	\(
		V[\ell](\tau) = \{v\in V: \tau[\ell,v] > \tau\}\,.
	\)
$\tau[\ell,v] = \frac{\underline{D}[\ell,v]}{(1+1/\eps)\La_{\max}}$ is a sufficient $\tau$-value for $\overline{\delta}_k[\ell,v]$ being $(1+\epsilon)$-upper-approximation of $D[\ell,v]$ within $I_k = [(k-1)\tau[\ell,v] , k\tau[\ell,v])$ (cf. Lemma~\ref{lemma:hole-radius-of-trapezoidal}).
The next theorem proves that $\alg{TRAP}$ provides a $(1+\eps)$-upper-approximation $\overline{\Delta}[\ell,v]$ for all faraway destinations from $\ell$, and also estimates the preprocessing requirements of the algorithm.
\begin{theorem}
\label{thm:trapezoidal-approximation-guarantee}
Fix $\ell\in L$, $F > f(n)$, and $\tau\in (0,T)$ s.t. $|V[\ell](\tau)| = n-F$. Let $\tau^* = \min_{v\in V[\ell](\tau)}\left\{\frac{\underline{D}[\ell,v]}{(1+1/\eps)\La_{\max}} \right\}$.
$\forall v\in V[\ell](\tau)$, $\overline{\Delta}[\ell,v]$ is the concatenation of all the upper-approximating functions $\overline{\de}_k[\ell,v]$ that $\alg{TRAP}$ returns per subinterval
	\(
		I_k = [~ t_{s_k}=(k-1)\tau^* ~,~ t_{f_k}=\min\{k\tau^*,T\} ~) : k\in\left[\ceil{\frac{T}{\tau^*}}\right]\,.
	\)
Then, $\forall v\in V[\ell](\tau)$,  $\overline{\De}[\ell,v]$ is a $(1+\eps)$-upper-approximation of $D[\ell,v]$ in $[0,T)$, requiring 
at most $2\ceil{\frac{T}{\tau^*}}$ breakpoints. 
The number of calls to $TDSP(\ell,\star,t)$ for their construction is 	
	\(
		\ceil{\frac{T}{\tau^*}} 
		\leq 1 + \frac{T(1+1/\eps) \La_{\max}}{\min_{v\in V[\ell](\tau)}\{\underline{D}[\ell,v]\}}
		\in \Order{n^a}\,.
	\)
\end{theorem}
\begin{proof}[Proof of Theorem~\ref{thm:trapezoidal-approximation-guarantee}]

$\tau^*$ is the appropriate length for the subintervals which assures that $\alg{TRAP}$ returns $(1+\eps)$-upper-approximations for all faraway destinations from $\ell$. By definition it holds that $\tau^*\geq \tau$.
Since $F > f(n)$, it holds that $\alg{TRAP}$ does not consider destinations at free-flow distance less than $1$. To see this, fix any $v\in V$ s.t. $\underline{D}[\ell,v]\leq 1$. By Assumption~\ref{assumption:Travel-Time-vs-Dijkstra-Rank},
\(
	\underline{\Gamma}[\ell,v] 
	\leq f(n)\cdot \underline{D}[\ell,v]^{\lambda} 
	\leq f(n) < F\,.
\)
Thus, we can be sure that $v\notin V[\ell](\tau)$. 
Since $T = n^a$, we conclude that $\frac{T}{\tau^*} = \frac{T(1+1/\eps)\La_{\max}}{\min_{v\in V[\ell](\tau)} \underline{D}[\ell,v]} \in \Order{n^a}$.
We proceed now with the analysis of $\alg{TRAP}$. $[0,T)$ is split into $\ceil{\frac{T}{\tau^*}}$ consecutive length-$\tau^*$ subintervals.
Lemma~\ref{lemma:upper+lower-approximating-functions-of-TRAP} assures that for each $I_k = [k\tau^*,(k+1)\tau^*)$ an upper-approximating function $\overline{\delta}_k[\ell,v]$ of $D[\ell,v]$ is determined, for each $v\in V[\ell](\tau)$.
The concatenation of all these functions constitutes the upper-approximating function $\overline{\Delta}[\ell,v]$ for $D[\ell,v]$ within $[0,T)$.
Since $\tau[\ell,v] \geq \tau^* \Rightarrow \underline{D}[\ell,v] \geq \left(1+\frac{1}{\eps}\right) \La_{\max} \tau^*$, we deduce (cf. Lemma~\ref{lemma:hole-radius-of-trapezoidal}) that, for all $v\in V[\ell](\tau)$, the produced upper-approximations within the consecutive length-$\tau^*$ intervals are  $(1+\eps)$-approximations of $D[\ell,v]$.
$\alg{TRAP}$ preprocesses $\ell\in L$ (concurrently for all $v\in V[\ell](\tau)$) by making $\ceil{\frac{T}{\tau^*}}\in \Order{n^a}$ calls to $TDSP(\ell,\star,t)$, to sample the endpoints of all the $\ceil{\frac{T}{\tau^*}}$ length-$\tau^*$ subintervals. For storing $\overline{\Delta}[\ell,v]$, it needs at most $2\ceil{\frac{T}{\tau^*}}$ breakpoints (there is at most one intermediate breakpoint $(\overline{t}_m,\overline{D}_m)$ per subinterval).
\end{proof}

\section{Oracles with fully-informed landmarks}
\label{section:FLAT-oracle}


In this section we describe two novel oracles, with landmarks possessing summaries for all reachable destinations, excluding possibly a small neighborhood around them. 
We start with a random landmark set $L \subset_{\uar(\rho)} V$, i.e., we decide independently and uniformly at random whether each vertex is a landmark, with probability $\rho = n^{-\omega}$ for a constant $\omega \in (0,1)$.
We consider as \term{faraway vertices} from $\ell\in L$, all the vertices at free-flow distance at most $\underline{R} = T^{\theta}$ from it, for a constant $\theta\in (0,1)$ to be determined later. $F = \max_{\ell\in L}\left\{ |\underline{B}[\ell;\underline{R}]| \right\}$ is the maximum number of faraway vertices from a landmark.
The next lemma shows that the main parameters we should consider w.r.t. a TD-instance are $\la$ (cf.~Assumption~\ref{assumption:Travel-Time-vs-Dijkstra-Rank}) and $a\in(0,1)$ s.t. $T = n^a$. All the other parameters essentially adjust their values to them.
\begin{lemma}
\label{lemma:period-vs-network-size}
	For $\nu\in(0,1)$ s.t. $T = diam(G,\underline{D})^{1/\nu}$, $\theta \in (0,1)$ s.t. $\frac{\nu}{\theta}\in \Order{1}$ and $\la,f,g$ defined as in Assumption~\ref{assumption:Travel-Time-vs-Dijkstra-Rank}, the following hold:
	(i) $\frac{1}{\la\nu} = \alpha \pm \order{1}$, and
	(ii) $F \in n^{\left[1\pm\order{1}\right]\theta/\nu}$.
\end{lemma}
%

\smallskip
\noindent
\textbf{\em The $\alg{TRAPONLY}$ oracle.}
A first attempt towards avoiding the dependency of the preprocessing requirements on $K^*$, is to develop an oracle whose preprocessing is based solely on $\alg{TRAP}$. The preprocessing of this oracle (we call it $\alg{TRAPONLY}$) first splits the entire period $[0,T)$ into consecutive subintervals of length $\tau = \frac{\underline{R}}{(1+1/\eps)\La_{\max}} > 0$ each. It then calls $\alg{TRAP}$ for each landmark $\ell\in L$, which guarantees $(1+\eps)$-upper-approximations for all the faraway destinations $v\in V[\ell](\tau)$ (cf. Theorem~\ref{thm:trapezoidal-approximation-guarantee}). As for the faraway destinations from $\ell$, we let their distances from $\ell$ be computed by the query algorithm that we are using (by growing a $\alg{TDD}$ ball from $\ell$). In particular, the query algorithm of $\alg{TRAPONLY}$ is an appropriate variant of $\alg{RQA}$,
we call it $\alg{RQA}^+$, which additionally grows a small $\alg{TDD}$ ball of size $F\polylog(F)$ (cf. Assumption~\ref{assumption:growth-of-free-flow-ball-sizes}) in order to compute the actual travel-times towards their faraway destinations. 
The following theorem analyzes the performance of $\alg{TRAPONLY}$.
\begin{theorem}
\label{thm:TRAPONLY-Complexities}
	The expected time of $\alg{RQA}^+$ and the preprocessing requirements of $\alg{TRAPONLY}$ are:
\(
	\Exp{Q_{\alg{RQA}^+}}								
	\in \Order{ n^{\omega r + \max\left\{ \omega , \frac{\theta}{\nu} \right\} + \order{1} } }
\) and
$S_{\alg{TRAPONLY}},~ P_{\alg{TRAPONLY}}
	\in \Order{ n^{2 + \alpha\cdot(1-\theta) - \omega + \order{1}} }$.
\end{theorem}
\begin{proof}[Proof of Theorem~\ref{thm:TRAPONLY-Complexities}]

During the preprocessing, $\alg{TRAPONLY}$ makes 
	\(
		\ceil{\frac{T}{\tau}}
		\leq 1 + \frac{T (1+ 1/\eps) \La_{\max}}{\underline{R}}
		= 1 + T^{1-\theta}(1+ 1/\eps) \La_{\max}
	\)
calls of $\alg{TDD}(\ell,t)$, for departure-times $t\in \left\{ 0, \tau, 2\tau, \ldots, \ceil{\frac{T}{\tau}}-1\right\}$ and landmarks $\ell\in L$, where the equality comes from Lemma~\ref{lemma:period-vs-network-size}.
Therefore, the preprocessing-time is dominated by the aggregate time for all these $\alg{TDD}$ probes. Taking into account that each $\alg{TDD}$ probe takes time $\Order{n \log(n) \log\log(K_{\max})}$ and that $|L| = \rho n = n^{1-\omega}$ landmarks, by using Lemma~\ref{lemma:period-vs-network-size} we get the following:
	\(
	P_{\alg{TRAPONLY}}
	= n^{1-\omega} \cdot n^{\frac{1-\theta}{\nu\la}[1+\order{1}]} \cdot n\log(n)\log\log(n)
	\in n^{ 2 - \omega + \frac{1-\theta}{\nu\la}[1+\order{1}] + \frac{\log\log(n)+\log\log\log(n)}{\log(n)} }
	= n^{ 2 - \omega + a\cdot(1-\theta) + \order{1} }\,.
	\)
The calculations are analogous for the required preprocessing space: For all landmarks $\ell\in L$ and all their faraway destinations $v\in V[\ell](\tau)$, the total number of breakpoints to store is at most
\(
	S_{\alg{TRAPONLY}}
	= 2\ceil{\frac{T}{\tau}}\rho n^2
	\in 		n^{2 -\omega + \frac{1-\theta}{\nu\la}[1+\order{1}] + \order{1}}	
	=			n^{2 -\omega + a\cdot(1-\theta) + \order{1}}\,.	
\)
As for the query-time complexity of $\alg{RQA}^+$, recall that the expected number of $\alg{TDD}$ balls that it grows is $\left( 1 / \rho \right)^{r}$. Additionally, $\alg{RQA}^+$ grows $\left( 1 / \rho \right)^{r}$ $\alg{TDD}$ balls from the corresponding closest landmarks. Each ball from a new center costs $\Order{(1/\rho)\log(1/\rho)}$. Each ball from a landmark costs $\Order{F\polylog(F)} \in n^{ [ 1 \pm \order{1}] \theta /\nu }$. Thus, the expected query-time is upper-bounded as follows:
	\(
	\Exp{Q_{\alg{RQA}^+}} \in \Order{(1/\rho)^r [(1/\rho)\log(1/\rho) + F\polylog(F)]\log\log(K_{\max})}
	= \Order{n^{ \omega r + \max\left\{ \omega , [1+\order{1}]\theta / \nu \right\} } }\,.
	\)
\end{proof}

The next corollaries are parameter-tuning examples showcasing the trade-offs among the sublinearity of query-time, the subquadratic preprocessing requirements and the stretch.
\begin{corollary}
\label{cor:TRAPONLY-Efficient-Complexities_scaling-IMPLIES-stretch}
For $\delta\in(\alpha,1)$, $\beta \in (0,\alpha^2\nu]$, $\omega = \frac{\delta}{r+1}$, $\theta = \frac{\delta\nu}{r+1}$ and
$r =  \floor{\frac{\delta\cdot(1 + \alpha\nu)}{\alpha + \beta}} - 1$, 
$S_{\alg{TRAPONLY}} , P_{\alg{TRAPONLY}}\in n^{2-\beta + \order{1}}$,
$\Exp{Q_{\alg{RQA}^+}} \in n^{\delta + \order{1}}$
and the stretch is
$1+\eps\cdot\frac{(1+\eps/\psi)^{r+1}}{(1+\eps/\psi)^{r+1} - 1}$.
\end{corollary}

\begin{corollary}
\label{cor:TRAPONLY-Efficient-Complexities_stretch-IMPLIES-scaling}
For any integer $k\geq 2$, let $\eta(k) = \ceil{\frac{ \log(k/(k-1)) }{\log(1+\eps/\psi)}} -1$, $\delta\in(0,1)$ and
$\beta \in \left( 0 , \frac{\delta}{\eta(k) + 2} \right)$. Then $\alg{TRAPONLY}$ achieves  stretch $1 + k\cdot\eps$ with 
$S_{\alg{TRAPONLY}} , P_{\alg{TRAPONLY}} \in n^{2-\beta + \order{1}}$ and 
$\Exp{Q_{\alg{RQA}^+}} \in n^{\delta + \order{1}}$, 
by scaling the TD-instance so that $T = n^{\alpha}$, for 
$\alpha = \frac{\delta - [\eta(k) + 2]\cdot\beta}{\eta(k) + 2 -\delta\nu}$.
\end{corollary}

\smallskip
\noindent
\textbf{\em The $\alg{FLAT}$ oracle.}
Our second attempt, the $\alg{FLAT}$ oracle, provides preprocessed information for all reachable destinations from each landmark. 
In particular, it uses the query algorithm $\alg{RQA}$ \cite{2014-Kontogiannis-Zaroliagis}
The preprocessing phase of $\alg{FLAT}$ is based on a proper combination of $\alg{BIS}$ and $\alg{TRAP}$ for constructing travel-time summaries. Each landmark $\ell\in L$ possesses summaries for all reachable destinations: $\alg{BIS}$ handles all the (at most $F = \max_{\ell\in L}\left\{ |\underline{B}[\ell;\underline{R}]| \right\}$) nearby destinations in $\underline{B}[\ell;\underline{R}]$, whereas $\alg{TRAP}$ handles all the faraway destinations of $V\setminus \underline{B}[\ell;\underline{R}]$. The space requirements for the summaries created by $\alg{TRAP}$ are exactly the same as in $\alg{TRAPONLY}$. As for the summaries computed by $\alg{BIS}$, we avoid the linear dependence of $\alg{BIS}$ on $K^*$ by assuring that $F$ is sufficiently small (but not too small) and exploiting Assumption~\ref{assumption:growth-of-free-flow-ball-sizes} which guarantees that the involved subgraph $\underline{B}'[\ell;F]$ in the preprocessing phase of $\alg{BIS}$ on behalf of $\ell$ has size $\Order{F\polylog(F)}$.
The next lemma shows exactly that $\alg{BIS}$ is only affected by the concavity-spoiling breakpoints of arc-travel-time functions in $\underline{B}'[\ell;F]$, rather than the entire graph.
\begin{lemma}
\label{lemma:irrelevance-of-faraway-vertices-in-travel-times-of-nearby-vertices}
\(
	\forall (\ell,v)\in L\times\underline{B}[\ell;F], \forall u\in V\setminus\underline{B}'[\ell;F], \forall t\in [0,T),
		D[\ell,v](t) < D[\ell,u](t)\,.
\)
\end{lemma}
\begin{proof}[Proof of Lemma~\ref{lemma:irrelevance-of-faraway-vertices-in-travel-times-of-nearby-vertices}]

From the definitions of the involved free-flow and full-congestion Dijkstra balls, the following holds:
\(
	D[\ell,v](t) \leq \upperD[\ell,v] \leq \overline{R}[\ell]
	< \lowerD[\ell,u] \leq D[\ell,u](t)\,.
\)
\end{proof}
The following theorem summarizes the complexities of the $\alg{FLAT}$ oracle.	
\begin{theorem}
\label{thm:FLAT-Complexities}
	The query-time $Q_{\alg{RQA}}$ and the preprocessing time $P_{\alg{FLAT}}$ and space $S_{\alg{FLAT}}$ of $\alg{FLAT}$ are:
	$\Exp{Q_{\alg{RQA}}}
		\in \Order{  n^{\omega(r+1) + \order{1}} }$ and
	$P_{\alg{FLAT}} ~,~ S_{\alg{FLAT}}
		\in \Order{  n^{1-\omega + \order{1}} \cdot [n^{2\theta/\nu} + n^{1 + \alpha\cdot(1-\theta)} ] }$.
\end{theorem}

\begin{proof}[Proof of Theorem~\ref{thm:FLAT-Complexities}]
$\alg{BIS}$ requires space at most $F^2\polylog(F)$, since by Lemma~\ref{lemma:irrelevance-of-faraway-vertices-in-travel-times-of-nearby-vertices} the involved graph only contains $F\polylog(F)$ vertices and concavity-spoiling breakpoints at the arc-travel-time functions. For the faraway vertices of $V\setminus \underline{B}[\ell;F]$, by setting $\tau = \frac{\underline{R}}{(1+1/\eps)\La_{\max}}$, $\alg{TRAP}$ provides $(1+\eps)$-approximate summaries, since $\forall v\in V\setminus \underline{B}[\ell;\underline{R}]$ the sufficient condition of Theorem~\ref{thm:trapezoidal-approximation-guarantee} holds:
\(
	\underline{D}[\ell,v]
	> \underline{R}
	= \left( 1 + 1/\eps \right)\La_{\max} \tau \,.
\)
Thus, 	we conclude that
	$S_{\alg{FLAT}}
			\in \rho n \left[ F^2\polylog(F) +
			\frac{T(1+1/\eps)\La_{\max} n}{\underline{R}} \right]
	$		
	$\DueTo{=}{L.\ref{lemma:period-vs-network-size}}
			n^{1-\omega} [n^{ (2\theta / \nu)\cdot[1+\order{1}] } + n^{ 1 + \alpha\cdot(1-\theta)[1+\order{1}] } ]
	$
	$= n^{1 - \omega	+ [1+\order{1}]\cdot\max\left\{~ 	2\theta / \nu ~,~ 1 + \alpha(1-\theta) ~\right\} + \order{1} }\,,
	$ 
since $f(n), g(n) \in \polylog(n)$.
\end{proof}
The next corollaries are parameter-tuning examples to showcase the effectiveness of $\alg{FLAT}$.
\begin{corollary}
\label{cor:FLAT-Parameter-Tuning_scaling-IMPLIES-stretch}
If $\delta\in(\alpha,1)$, $\beta \in \left( 0 , \frac{\alpha\cdot(1+\alpha)}{2 / \nu + \alpha} \right]$, 
$\omega = \frac{\delta}{r+1}$,
$r = \floor{\frac{\delta}{\alpha}\cdot\frac{ 2 / \nu + \alpha}{(\beta / \alpha)\cdot\left( 2 / \nu + \alpha \right) + \left( 2 / \nu - 1\right)}} - 1$ and
$\theta = \frac{1+\alpha}{ 2 / \nu + \alpha}$, then $\alg{FLAT}$ has 
$P_{\alg{FLAT}} , S_{\alg{FLAT}} \in n^{2- \beta + \order{1}}$,
$\Exp{Q_{\alg{RQA}}} \in n^{\delta + \order{1}}$
and stretch $1+\eps\cdot\frac{(1+\eps/\psi)^{r+1}}{(1+\eps/\psi)^{r+1} - 1}$.
\end{corollary}
\begin{corollary}
\label{cor:FLAT-Efficient-Complexities_stretch-IMPLIES-scaling}
For integer $k\geq 2$, let $\eta(k) = \ceil{\frac{ \log(k/(k-1)) }{\log(1+\eps/\psi)}} -1$ and $\delta\in(0,1)$. $\alg{FLAT}$ achieves a target stretch $1 + k\cdot\eps$ with preprocessing requirements $n^{2 - \order{1}}$ and expected query-time $n^{\delta + \order{1}}$, by scaling the TD-instance so that $T = n^{\alpha}$ for $\alpha = \frac{2\delta}{[\eta(k) + 2]\cdot(2-\nu) -\delta\nu}$, as $\beta\downarrow 0$.
\end{corollary}

\smallskip
\noindent
\textbf{\em Comparison of $\alg{TRAPONLY}$ and $\alg{FLAT}$.}
Both $\alg{TRAPONLY}$ and $\alg{FLAT}$ depend on the travel-time metric, but are independent of the degree of disconcavity $K^*$. 
On one hand, $\alg{TRAPONLY}$ is a simpler oracle, at least w.r.t. its preprocessing phase.
On the other hand, $\alg{FLAT}$ achieves a better approximation for the same TD-instance and anticipations for sublinear query-time $n^{\delta}$ and subquadratic preprocessing requirements $n^{2-\beta}$. This is because, as $\beta\downarrow 0$, $\alg{FLAT}$ guarantees a recursion budget $r$ of (roughly) $\frac{2\delta}{a} - 1$, whereas $\alg{TRAPONLY}$ achieves about half this value and $r$ has an exponential effect on the stretch that the query algorithms achieve.

\section{The $\alg{HORN}$ oracle}
\label{section:HORN-oracle}


	%
\begin{wrapfigure}{r}{6cm}
\centerline{\includegraphics[width=6cm]{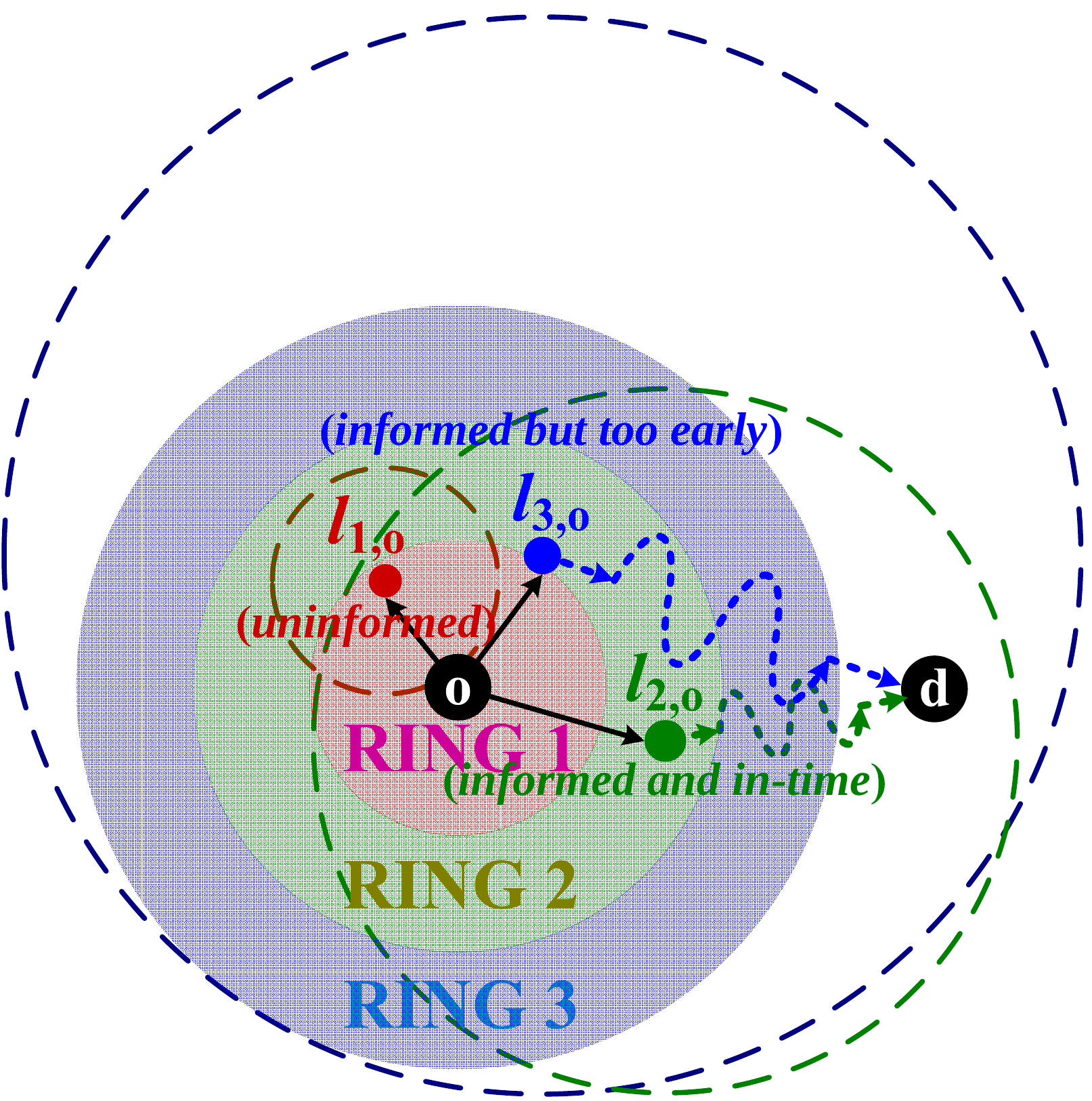}}
\caption{\label{fig:HQA-execution}
	Demonstration of execution of $\alg{HQA}$. Dashed circles indicate areas of coverage. Solid circular stripes indicate the rings of the corresponding levels in the hierarchy. Landmark $\ell_{1,o}$ is uninformed and $\ell_{3,o}$, although informed, comes too early. $\ell_{2,o}$ is both informed and within the ring of its own level, leading $\alg{HQA}$ to deduce that the appropriate level is $i=2$.
}
\end{wrapfigure}

We now describe and analyze the \term{Hierarchical ORacle for time-dependent Networks} ($\alg{HORN}$), whose query algorithm is highly competitive against $\alg{TDD}$, not only for long-range queries (i.e., having Dijkstra-Rank proportional to the network size) but also for medium- and short-range queries, while ensuring \emph{subquadratic} preprocessing space and time.
%
The main idea of $\alg{HORN}$ is to preprocess: many landmarks, each possessing summaries for a few destinations around them, so that all short-range queries can be answered using only these landmarks; fewer landmarks possessing summaries for more (but still not all) destinations around them, so that medium-range queries be answered by them; and so on, up to only a few landmarks (those required by $\alg{FLAT}$) possessing summaries for all reachable destinations.  
The \term{area of coverage} $C[\ell]\subset V$ of $\ell$ is the set of its nearby vertices, for which $\ell$ possesses summaries. $\ell$ is called \term{informed} for each $v\in C[\ell]$, and \term{uninformed} for each $v\in V\setminus C[\ell]$. The landmarks are organized in a hierarchy, according to the sizes of their areas of coverage. Each level $L_i$ in the hierarchy is accompanied with a \term{targeted Dijkstra-Rank} $N_i\in[n]$, and the goal of $\alg{HORN}$ is to assure that $L_i$ should suffice for $\alg{RQA}$ to successfully address queries $(o,d,t_o)$ with $\Gamma[o,d](t_o) \leq N_i$, in time $\order{N_i}$.
The difficulty of this approach lies in the analysis of the query algorithm. We want to execute a variant of $\alg{RQA}$ which, based on a \emph{minimal} subset of landmarks, would guarantee a $(1+\s(r))$-approximate solution for any query $(o,d,t_o)$ (as in $\alg{FLAT}$), but also time-complexity sublinear in $\G[o,d](t_o)$. 
We propose the \term{Hierarchical Query Algorithm} ($\alg{HQA}$) which grows an initial ball from $(o,t_o)$ that stops only when it settles an informed landmark $\ell$ w.r.t. $d$ which is at the ``right distance'' from $o$, given the density of landmarks belonging to the same level with $\ell$. 
$\alg{HQA}$ essentially ``guesses'' as appropriate level-$i$ in the hierarchy the level that contains $\ell$, and continues with the execution of $\alg{RQA}$ with landmarks having coverage at least equal to that of $\ell$ (cf. Fig.~\ref{fig:HQA-execution}).

\smallskip
\noindent
\emph{\textbf{Description of $\alg{HORN}$.}}
We use the following parameters for the hierarchical construction:
(i) $k\in\Order{\log\log(n)}$ determines the number of levels (minus one) comprising the hierarchy of landmarks.
(ii) $\gamma > 1$ determines the actual values of the targeted Dijkstra-Ranks, one per level of the hierarchy. In particular, as $\gamma$ gets closer to $1$, the targeted Dijkstra-Ranks accumulate closer to small- and medium-rank queries.
(iii) $\delta\in (0,1)$ is the parameter that quantifies the sublinearity of the query algorithm ($\alg{HQA}$), in each level of the hierarchy, compared to the targeted Dijkstra-Rank of this level. In particular, if $N_i$ is the targeted Dijkstra-Rank corresponding to level-$i$ in the hierarchy, then $\alg{HQA}$  should be executed in time $\Order{(N_i)^{\delta}}$, if only the landmarks in this level (or in higher levels) are allowed to be used.

\smallskip
\noindent
\emph{Preprocessing of $\alg{HORN}$.}
$\forall i\in[k]$, 
%
	we set the targeted Dijkstra-Rank for level-$i$ to $N_i = n^{(\g^i-1)/\g^i}$.
Then, we construct a randomly chosen level-$i$ landmark set $L_i \subset_{\uar(\rho_i)} V$, where 
$\rho_i = N_i^{- \delta / (r+1)} = n^{- \delta(\g^i-1) / [(r+1)\g^i]}$.
Each $\ell_i\in L_i$ acquires summaries for all (and only those) $v\in C[\ell_i]$, where $C[\ell_i]$ is the smallest \emph{free-flow} ball centered at $\ell_i$ containing $c_i = N_i\cdot n^{\xi_i} = n^{(\g^i-1)/\g^i + \xi_i}$ vertices, for a sufficiently small constant $\xi_i>0$.
The summaries to the $F_i = c_i^{\chi}$ nearby vertices around $\ell_i$ are constructed with $\alg{BIS}$; the summaries to the remaining $c_i - F_i$ faraway vertices of $\ell_i$ are constructed with $\alg{TRAP}$, where $\chi = \frac{\theta}{\nu} = \frac{1 + \alpha}{2 + \alpha\nu}\in \left[\frac{1}{2},\frac{2}{2 + \nu}\right]$ is an appropriate value determined to assure the correctness of $\alg{FLAT}$ w.r.t. the level-$i$ of the hierarchy.
%
An ultimate level $L_{k+1}\subset_{\uar(\rho_{k+1})} V$ of landmarks, with $\rho_{k+1} = n^{-\frac{\delta}{r+1}}$, assures that $\alg{HORN}$ is also competitive against queries with Dijkstra-Rank greater than $n^{(\g^k-1)/\g^k}$. We choose in this case $c_{k+1} = N_{k+1} = n$, $F_{k+1} = n^{\chi}$ and $C[\ell_{k+1}] = V$, $\forall \ell_{k+1} \in L_{k+1}$.

\smallskip
\noindent
\emph{Description of $\alg{HQA}$.}
A $\alg{TDD}$ ball from $(o,t_o)$ is grown until $d$ is settled, or the \emph{(ESC)}-criterion or the \emph{(ALH)}-criterion is fulfilled (whichever occurs first):
\newline
~$\diamond$~
\textbf{Early Stopping Criterion (ESC):} $\ell_o\in L = \union_{i\in[k+1]} L_i$ is settled, which is informed ($d\in C[\ell_o]$) and, for $\varphi\geq 1$, 
\(
	\frac{\overline{\De}[\ell_o,d](t_o+D[o,\ell_o](t_o))}{D[o,\ell_o](t_o)}
	\geq (1+\eps)\cdot\varphi\cdot(r+1) + \psi - 1\,.
\)
%
\newline
~$\diamond$~
\textbf{Appropriate Level of Hierarchy (ALH):}
	For some level $i\in[k]$ of the hierarchy, the first landmark $\ell_{i,o}\in L_i$ is settled such that:
	\emph{(i)} $d\in C[\ell_{i,o}]$ ($\ell_{i,o}$ is ``informed''); and
	\emph{(ii)} 
		\(
			\frac{N_i^{\delta / (r+1) }}{\ln(n)}
			\leq \Gamma[o,\ell_{i,o}](t_o)
			\leq \ln(n)\cdot N_i^{\delta / (r+1)} 
		\) 
		($\ell_{i,o}$ is at the ``right distance'').
In that case, $\alg{HQA}$ concludes that $i$ is the ``appropriate level'' of the hierarchy to consider. Observe that the level-$(k+1)$ landmarks are always informed. Thus, if no level-$(\leq k)$ informed landmark is discovered at the right distance, then the first level-$(k+1)$ landmark that will be found at distance larger than $\ln(n)\cdot N_k^{ \delta /(r+1) }$ will be considered to be at the right distance, and then $\alg{HQA}$ concludes that the appropriate level is $k+1$.

If $d$ is settled, an exact solution is returned.
If (ESC) causes termination of $\alg{HQA}$, the value $D[o,\ell_o](t_o) + \overline{\De}[\ell_o,d](t_o+D[o,\ell_o](t_o))$ is reported.
Otherwise, $\alg{HQA}$ stops the initial ball due to the (ALH)-criterion, considering $i\in[k+1]$ as the appropriate level, and then continues executing the variant of $\alg{RQA}$, call it $\alg{RQA}_i$, which uses as its landmark set $M_i = \union_{j=i}^{k+1} L_j$.
Observe that $\alg{RQA}_i$ may fail constructing approximate solutions via certain landmarks in $M_i$ that it settles, since they may not be informed about $d$. 
Eventually, $\alg{HQA}$ returns the best $od$-path (w.r.t. the approximate travel-times) among the ones discovered by $\alg{RQA}_i$ via \emph{all} settled and informed landmarks $\ell$.
Theorem~\ref{thm:HORN-Complexities} summarizes the performance of $\alg{HORN}$.
\begin{theorem}
\label{thm:HORN-Complexities}
Consider any TD-instance with $\la\in \order{\sqrt{\frac{\log(n)}{\log\log(n)}}}$ and
$g(n), f(n)\in \polylog(n)$ (cf. Assumption~\ref{assumption:Travel-Time-vs-Dijkstra-Rank}).
For $\varphi = \frac{\eps\cdot(r+1)}{\psi\cdot(1+\eps/\psi)^{r+1} - 1}$ and
$k\in\Order{\log\log(n)}$, let
$\xi_i \in\left( \left[(1+\la)\cdot \log\log(n) + \la\log\left(1 + \frac{\zeta}{1-\La_{\min}}\right)\right]  / \log(n) ~,~ 1-\g^{-i}\right)$, for all $i\in[k]$. 	
Then, for any query $(o,d,t_o)$ s.t. $N_{i^* - 1} < \G[o,d](t_o) \leq N_{i^*}$ for some $i^*\in[k+1]$,
any $\delta \in (a,1)$, $\beta>0$, and
$r = \floor{\frac{\delta}{a} \cdot \frac{(2 / \nu + a)(1 - \g)}{\beta \cdot (2/(a\nu) + 1) + 2 / \nu - 1}} - 1$,
$\alg{HORN}$ achieves
$\Exp{Q_{\alg{HQA}}} \in (N_{i^*})^{\delta+\order{1}}$,
$P_{\alg{HORN}} ~,~ S_{\alg{HORN}} \in n^{2-\beta+\order{1}}$ and
stretch $1+\eps\frac{(1+\eps/\psi)^{r+1}}{(1+\eps/\psi)^{r+1} -1 }$, 
with probability at least $1-\Order{\frac{1}{n}}$.
\end{theorem}

\bibliographystyle{plain}


\clearpage

\appendix


\section{Summary of Notation}
\label{section:notation}

\begin{center}
\begin{footnotesize}
\begin{tabular}{|c||p{10.5cm}|}
\hline
Symbol & Description
\\ \hline\hline
$[n]$ & The set of integers $\{1,2,\ldots,n\}$, for any $n\in\naturals-\{0\}$.
\\ \hline
$G=(V,A)$ & The graph representing the underlying structure of the road network.
	\newline
	$n = |V|$ and $m=|A|$.
\\ \hline
$diam(G,D)$ & The diameter of $G$ under an arc-cost metric $D$.
\\ \hline
$\mathcal{P}_{o,d}$ & Set of $od$-paths in $G$.
\\ \hline
$p \bullet q$ & The concatenation of the $ux$-path $p$ with the $xv$-path $q$ at vertex $x$.
\\ \hline
$ASP[o,d](t_o)$ & Set of $(1+\eps)$-approximations of minimum-travel-time $od$-paths in $G$, for given departure-time $t_o\geq 0$.
\\ \hline
$SP[o,d](t_o)$ & Set of minimum-travel-time $od$-paths in $G$, for given departure-time $t_o\geq 0$.
\\ \hline
$B[v](t_v)$ & A ball growing from $(v,t_v)\in V\times[0,T)$, in the time-dependent metric, until either the destination $d$ or the closest landmark $\ell_v$ from $(v,t_v)$ is settled.
\\ \hline
$B[v;F](t_v)$ & A ball growing from $(v,t_v)\in V\times[0,T)$, in the time-dependent metric, of size $F\in\naturals$.
\\ \hline
$\overline{B}[v;F]$ / $\underline{B}[v;F]$ & A ball growing from $v\in V$, in the full-congestion / free flow metric, of (integer) size $F\in\naturals$.
\\ \hline
$\overline{B}[v;R]$ / $\underline{B}[v;R]$ & A ball growing from $v\in V$, in the full-congestion / free flow metric, of (scalar) radius $R>0$.
\\ \hline
$B'[v;F](t_v)$ & A ball growing from $(v,t_v)\in V\times[0,T)$, in the time-dependent metric, of size $F\polylog(F)$, according to Assumption~\ref{assumption:growth-of-free-flow-ball-sizes}.
\\ \hline
$d[a](t)$ & The limited-window arc-travel-time function for arc $a\in A$, with departure-time $t\in [0,T)$ for some \emph{constant} time-period $T>0$ (e.g., a single day).
\\ \hline
$M_a$ & Maximum possible travel-time ever seen at arc $a$.
\\ \hline
$M$ & Maximum arc-travel-time ever seen in any arc.
\\ \hline
$D[a](t)$ & Periodic arc-travel-time function for arc $a\in A$, with domain $t\in [0,\infty)$.
\\ \hline
$Arr[a](t)$ & The arc-arrival-time function for arc $a\in A$.
\\ \hline
$D[o,d]$ & Minimum-travel-time \emph{function}, from $o$ to $d$.
\\ \hline
$\Gamma[o,d]$ & Dijkstra-Ranks \emph{function}, from $o$ to $d$.
\\ \hline
$D_{\max}[o,d]$ / $D_{\min}[o,d]$ & The maximum and minimum value of $D[o,d]$.
\\ \hline
$\overline{D}[a]$ / $\underline{D}[a]$ & Travel-times of $a$ in \emph{full-congestion} and \emph{free-flow} metrics, respectively.
\\ \hline
$\overline{\De}[o,d]$ / $\underline{\De}[o,d]$ & An upper-approximating / lower-approximating function to $D[o,d]$.
\\ \hline
$Arr[o,d]$ & Earliest-arrival-time \emph{function}, from $o$ to $d$.
\\ \hline
$t_u ~~ (t_v)$ & Departure-time from the tail $u$ (arrival-time at the head $v$) for $uv\in A$.
\\ \hline
$TDSP(o,d,t_o)$ & The problem of finding a min-cost $od$-path, given a departure-time $t_o$.
\\ \hline
$TDSP(o,\star,t_o)$ & The problem of finding a min-cost paths tree, for given departure-time $t_o$.
\\ \hline
$TDSP(o,d)$ & The problem of constructing a succinct representation of min-cost $od$-paths \emph{function}.
\\ \hline
$K_a$ & Number of breakpoints in the arc-travel-time function $D[a]$.
\\ \hline
$K$ & Total number of breakpoints in the arc-travel-time functions.
\\ \hline
$K_{\max}$ & The maximum number of breakpoints, among the arc-travel-time functions.
\\ \hline
$K^*$ & Total number of \emph{concavity-spoiling} breakpoints (i.e., points at which the slope increases) in the arc-travel-time functions.
\\ \hline
$\La_{\max}$ & Maximum slope among minimum-travel-time functions.
\\ \hline
$\La_{\min}$ & Absolute value of minimum slope among minimum-travel-time functions.
\\ \hline
$\zeta$ & ratio of minimum-travel-times in opposite directions between two vertices for any specific departure-time.
\\ \hline
$r$ & The recursion budget for $\alg{RQA}$ and $\alg{HQA}$.
\\ \hline
$\alg{BIS}$ & The bisection approximation method for minimum-travel-time functions.
\\ \hline
$\alg{TRAP}$ & The trapezoidal approximation method for minimum-travel-time  functions.
\\ \hline
$\alg{FCA}$ & The \emph{Forward Constant Approximation} query algorithm.
\\ \hline
$\alg{RQA}$ & The \emph{Recursive Query Algorithm}, based on landmarks possessing information towards all possible destinations.
\\ \hline
$\alg{HQA}$ & The \emph{Hierarchical Query Algorithm}, based on a hierarchy of landmarks.
\\ \hline
$\alg{FLAT}$ & The oracle that uses landmarks possessing summaries towards all possible destinations, and the $\alg{RQA}$ query algorithm.
\\ \hline
$\alg{HORN}$ & The oracle that uses a hierarchy of landmarks with their own subset of destination vertices, and the $\alg{HQA}$ query algorithm.
\\ \hline
\end{tabular}
\end{footnotesize}
\end{center}

\section{Assumptions on the travel-time metric}
\label{section:assumptions}

In this section, we make a few assumptions on the kind of minimum-travel-time functions in the network. All assumptions are quite natural and justified in several application scenarios, such as the urban-traffic road networks, which have motivated this work. Technically, they allow a \emph{smooth transition} from static metrics on undirected graphs towards time-dependent metrics on directed graphs.

The first assumption, called \emph{Bounded Travel-Time Slopes}, asserts that the partial derivatives of the minimum-travel-time functions between any pair of origin-destination vertices are bounded in a known fixed interval:
%
\\[5pt]\noindent
\textbf{Assumption~\ref{assumption:Bounded-Travel-Time-Slopes}.}
\emph{
	There exist constants $\La_{\min}\in[0,1)$ and $\La_{\max}\geq 0$ s.t.:
\(
	\forall (o,d)\in V\times V,~ \forall 0\leq t_1 < t_2,~
	\frac{D[o,d](t_1) - D[o,d](t_2)}{t_1 - t_2} \in[-\La_{\min}, \La_{\max}]\,.
\)}
\\[5pt]\noindent
%
The lower-bound of $-1$ in the minimum-travel-time function slopes is indeed a direct consequence of the FIFO property, which is typically assumed to hold in several time-dependent networks, such as road networks.
$\La_{\max}$ represents the maximum possible rate of change of minimum-travel-times in the network, which only makes sense to be bounded (in particular, independent of the network size) in realistic instances such as the ones representing urban-traffic time-dependent
road networks.

The second assumption, called \emph{Bounded Opposite Trips}, asserts that for any given departure time, the minimum-travel-time from $o$ to $d$ is not more than a  \emph{constant} $\zeta\geq 1$ times the minimum-travel-time in the opposite direction (but not necessarily along the reverse path).
%
\\[5pt]\noindent
\textbf{Assumption~\ref{assumption:Bounded-Opposite-Trips}.}
\emph{$\exists\zeta \geq 1$,
\(
	\forall (o,d)\in V\times V,~ \forall t\in[0,T),~
	D[o,d](t) \leq \zeta\cdot D[d,o](t)\,.
\)}
\\[5pt]\noindent
%
This is quite natural in road networks, (i.e., it is most unlikely that a trip in one direction is more than $10$ times longer than the trip in the opposite direction during the same time period).

\remove{
These assumptions were verified through an experimental analysis we conducted (see \cite{2014-Kontogiannis-Zaroliagis,2015-Kontogiannis-Michalopoulos-Papastavrou-Paraskevopoulos-Wagner-Zaroliagis}
for the details), using two data sets: a real-world time-dependent snapshot of two weeks traffic data of the city of Berlin, kindly provided to us by TomTom (consisting of $n = 478,989$ vertices and $m = 1,134,489$ arcs), and a benchmark time-dependent instance of Western Europe's (WE) road network (consisting of $n = 18,010,173$ vertices and $m = 42,188,664$ arcs) kindly provided by PTV AG for scientific use.
Our experimental analysis showed that for the Berlin data set, $\La_{\max} < 0.2$ and $\zeta < 1.6$, while for the WE benchmark $\La_{\max} < 6.2$ and $\zeta < 1.2$.
}

A third assumption concerns the relation of the Dijkstra ranks of cocentric balls in the network, with respect to the (static) \emph{free-flow} metric implied by the time-dependent instance at hand. Its purpose is to bridge the gap between expansion of graph distances (densities of Dijkstra balls) and travel-times in the network. It essentially asserts that, given a particular origin-vertex, if one considers a free-flow ball  of a certain size (and travel-time radius) and then expands further this ball to a larger radius (equal to the full-congestion radius in the free-flow ball) then the ball size that we get changes by at most a polylogarithmic factor.
\remove{
Also this assumption has been verified in the real time-dependent data for the city of Berlin that we have at our disposal (cf. \cite{2015-Kontogiannis-Michalopoulos-Papastavrou-Paraskevopoulos-Wagner-Zaroliagis}), where the ball-size expansion factor is bounded by a constant smaller than 2.
}
%
\\[5pt]\noindent
\textbf{Assumption~\ref{assumption:growth-of-free-flow-ball-sizes}.}
\emph{For any vertex $\ell\in V$, and a positive integer $F$, consider the (static) Dijkstra ball $\underline{B}[\ell;F]$ around $\ell$ under the free-flow metric.
Let
	$\underline{R}[\ell] = \max\{ \underline{D}[\ell,v]: v\in \underline{B}[\ell;F] \}$
and
	$\overline{R}[\ell] = \max\{ \overline{D}[\ell,v]: v\in \underline{B}[\ell;F] \}$
be the largest free-flow and full-congestion travel-times from $\ell$ to any other vertex in $\underline{B}[\ell;F]$.
Finally, let
	$\underline{B}'[\ell;F] = \{ v\in V: \underline{D}[\ell,v](0,T) \leq \overline{R}[\ell] \}$
be the free-flow ball around $\ell$ with the (larger) radius $\overline{R}[\ell]$.
Then it holds that $|\underline{B}'[\ell;F]| \in \Order{F\cdot\polylog(F)}$.}
\\[5pt] \noindent
The aforementioned assumptions were verified through an experimental analysis on two real-world road networks, one concerning the urban-area of the city of Berlin and the other concerning the national road network of Germany. Our experimental analysis, presented
in \cite{2016-Kontogiannis-Michalopoulos-Papastavrou-Paraskevopoulos-Wagner-Zaroliagis},
shows that for the Berlin data set, $\La_{\max} < 0.062$, $\zeta < 1.2$, and the maximum ball size expansion factor $\le 6.7$, while for the Germany data set $\La_{\max} < 0.22$, $\zeta < 1.1$, and the maximum ball size expansion factor $\le 8.3$.

%
Finally, we need a systematic way to correlate the arc-cost metric (travel-times)
with the Dijkstra-Rank metric induced by it. For this reason, inspired by the notion of the doubling dimension (e.g., \cite{2011-Bartal-Gottlieb-Kopelowitz-Lewenstein-Roditty} and references therein), we consider some \emph{scalar} $\lambda \geq 1$ and functions $f,g:\naturals\mapsto [1,\infty)$, such that the following hold:
$\forall (o,d,t_o)\in V\times V\times [0,T)$,
	(i) $\Gamma[o,d](t_o) \leq f(n) \cdot (D[o,d](t_o))^{\la}$, and
	(ii) $D[o,d](t_o) \leq g(n) \cdot (\Gamma[o,d](t_o))^{1/\la}$.
This property trivially holds, e.g., for $\la=1$, $f(n) = n$, and $g(n) = diam(G,\overline{D})$. Of course, our interest is for the smallest possible values of $\la$ and at the same time the slowest-growing functions $f(n),g(n)$.
Our last  assumption exactly quantifies the boundedness of this correlation by restricting $\la$, $f(n)$ and $g(n)$.
%
\\[5pt]\noindent
\textbf{Assumption~\ref{assumption:Travel-Time-vs-Dijkstra-Rank}.}
\emph{For the graph $G=(V,A)$ and the time-dependent arc-cost metric $D$ that we consider, it holds that there exist $\la\in\order{\frac{\log(n)}{\log\log(n)}}$} and $f(n),g(n) \in \polylog(n)$ such that
	(i) $\Gamma[o,d](t_o) \leq f(n) \cdot (D[o,d](t_o))^{\la}$, and
	(ii) $D[o,d](t_o) \leq g(n) \cdot (\Gamma[o,d](t_o))^{1/\la}$.
\\[5pt] \noindent
Note that static oracles related to the notion of doubling dimension (e.g., \cite{2011-Bartal-Gottlieb-Kopelowitz-Lewenstein-Roditty}), demand a \emph{constant} value for the exponent $\la$ of the expansion, whereas we allow even a (not too fast) growing function of the network size $n$. The notion of expansion that we consider introduces some additional slackness, by allowing some divergence from the corresponding powers by polylogarithmic factors of the network size.

\section{Review of results in \cite{2014-Kontogiannis-Zaroliagis}}
\label{section:2014-Kontogiannis-Zaroliagis-Overview}
%
\smallskip
\noindent
\emph{\textbf{Review of the approach in \cite{2014-Kontogiannis-Zaroliagis}.}} 
The TD-oracle in \cite{2014-Kontogiannis-Zaroliagis} starts by first determining, for some $\rho \in (0,1)$, a set $L$ of $\rho n$ independently and uniformly at random selected \emph{landmarks} (vertices acting as reference points). 
During the preprocessing phase, all $(1+\eps)$-upper-approximating functions (\term{travel-time summaries}) $\overline{\Delta}[\ell,v]$ are constructed from each landmark $\ell\in L$ towards every reachable destination $v\in V$, using the $\alg{BIS}$ approximation algorithm that keeps bisecting the common axis of departure-times from $\ell$, until the desired approximation guarantee is achieved in each subinterval, for all destinations.
It is proved in \cite{2014-Kontogiannis-Zaroliagis} that $\alg{BIS}$ requires 
\[
	\Order{ 	
				\frac{K^*}{\eps}
				\max_{d\in V}\left\{\log\left(\frac{T\cdot(\La_{\max}+1)}{\eps \min_{t_o}\{D[o,d](t_o)\}}\right)\right\}
				\max_{d\in V}\left\{\log\left(\frac{\max_{t_o}\{D[o,d](t_o)\}}{\min_{t_o}\{D[o,d](t_o)\}}\right)\right\}
	}
\]
	calls to $TDSP(o,\star,t_o)$,  for a given origin $o\in V$ and all reachable destinations from it. The following lemma clarifies this under the lens of our TD-instnace:
\begin{lemma}
\label{lemma:BIS-analysis}
	$\alg{BIS}$ requires $\Order{\frac{K^*\cdot\log^2(n)}{\eps}}$ calls to $TDSP(o,\star,t_o)$ to provide all the summaries of minimum-travel-time functions from a given origin $o\in V$ towards all destinations at distance at least $1$ from $o$.
\end{lemma}
\begin{proof}[Proof of Lemma~\ref{lemma:BIS-analysis}]
The crucial observation, which is a direct consequence of Assumption~\ref{assumption:Bounded-Travel-Time-Slopes}, is that:
\[
	\overline{D}[o,d] - \underline{D}[o,d] \leq \La_{\max}\cdot T
	\Rightarrow
	\frac{\overline{D}[o,d]}{\underline{D}[o,d]} \leq \frac{\La_{\max}\cdot T}{\underline{D}[o,d]} + 1
\]
Exploiting the facts that $T = n^a$, $\underline{D}[o,d]\geq 1$ and $\La_{\max}\in\Order{1}$, we conclude that $\alg{BIS}$ requires 
\(
	\Order{ 	
				\frac{K^*}{\eps}\cdot \log^2\left(\frac{n}{\eps}\right)
	}
\)
calls to $TDSP(o,\star,t_o)$.
\end{proof}
Two query algorithms were proposed, $\alg{FCA}$ and $\alg{RQA}$, which provide constant and $(1+\s)$-approximations (for \emph{constant} $\s>\eps$) to minimum-travel-times, respectively.

$\alg{FCA}$ is a simple \emph{sublinear}-time algorithm for evaluating $\overline{\Delta}[o,d](t_o)$, guaranteeing a constant approximation w.r.t. $D[o,d](t_o)$. In particular, it grows a $\alg{TDD}$ ball $B[o](t_o) = \left\{x\in V : D[o,x](t_o) \leq D[o,\ell_o](t_o)\right\}$ from $(o,t_o)$, until either $d$ or the closest landmark $\ell_o \in \arg\min_{\ell\in L}\{ D[o,\ell](t_o)\}$ is settled. $\alg{FCA}$ then returns either the exact travel-time value, or the approximate travel-time value via $\ell_o$, $\overline{\Delta}[o,d](t_o) = D[o,\ell_o](t_o) + \overline{\Delta}[\ell_o,d](t_o+D[o,\ell_o](t_o))$, which is a guaranteed $(1+\eps+\psi)$-approximation; $\psi$ is a constant depending on $\eps, \zeta$ and $\La_{\max}$, but not on the size of the network.
$\alg{RQA}$ improves the approximation guarantee provided by $\alg{FCA}$, by exploiting carefully a number of recursive accesses to the preprocessed information, each of which produces (via calls to $\alg{FCA}$) additional candidate $od$-paths. The tuning parameter $r\in \naturals$ -- the \term{recursion budget} -- is the depth of the produced recursion tree.
$\alg{RQA}$ works as follows: As long as the destination vertex has not yet been discovered in the explored area around the origin, and there is still some remaining recursion budget, it ``guesses'' (by exhaustively searching for it) the next vertex $w_k$  at the boundary of the current ball, along the (unknown) shortest $od$-path. Then, it grows a new $\alg{TDD}$ ball from the new center
\(
		\left(~ w_k ~,~ t_k = t_o+D[o,w_1](t_o) + D[w_1,w_2](t_1)+\cdots + D[w_{k-1},w_k](t_{k-1}) ~\right)
\)
until it reaches the closest landmark $\ell_k$ to it, at distance $R_k = D[w_k,\ell_k](t_k)$. $\ell_k$ offers an alternative $od$-path $SOL_k = P_{o,w_1}\bullet\cdots\bullet P_{w_{k-1},w_k} \bullet Q_k\bullet \Pi_k$ by a new application of $\alg{FCA}$, where $P_{w_i,w_{i+1}}\in SP[w_i,w_{i+1}](t_i)$, $Q_k\in SP[w_k,\ell_k](t_k)$, and $\Pi_k\in ASP[\ell_k,d](t_k+R_k)$ is the approximate suffix subpath provided by the oracle. Observe that $SOL_k$ uses a longer (optimal, if all centers lie on the unknown min-cost path) prefix-subpath $P_{o,w_1}\bullet\cdots\bullet P_{w_{k-1},w_k}$ which is then completed with a shorter approximate suffix-subpath $Q_k\bullet\Pi_k$.
It is proved in \cite{2014-Kontogiannis-Zaroliagis} that the minimum-travel-time over all the discovered approximate $od$-paths discovered by $\alg{RQA}$, is a $(1+\s)-$approximation of $D[o,d](t_o)$, for any constant $\s > \eps$.
The next theorem is a consequence of the analysis in \cite{2014-Kontogiannis-Zaroliagis}:

\begin{theorem}[\cite{2014-Kontogiannis-Zaroliagis}]
\label{thm:BIS+RQA-Complexities}
If a TD-instance with $m \in \Order{n}$ and compliant with Assumptions \ref{assumption:Bounded-Travel-Time-Slopes} and \ref{assumption:Bounded-Opposite-Trips}
is preprocessed using $\alg{BIS}$ for constructing travel-time summaries from $\rho n$ landmarks chosen uniformly-at-random,
then the expected values of preprocessing space $S_{\alg{BIS}}$ and time $P_{\alg{BIS}}$, and query time $Q_{\alg{RQA}}$ for $\alg{RQA}$, are:
	$\Exp{ S_{\alg{BIS}} }
	 \in \Order{ \rho n^2 (K^*+1) }$,
	$\Exp{ P_{\alg{BIS}} }
	\in \Order{ \rho n^2 (K^*+1) \log(n) \log\log(K_{\max}) }$,
	and
	$\Exp{ Q_{\alg{RQA}} }
	\in \Order{ \left( 1 / \rho \right)^{r+1}
			\log\left(  1 / \rho \right) \log\log(K_{\max}) }$, where $r\in\naturals$ is the recursion depth in $\alg{RQA}$ (for $r=0$ we get $\alg{FCA}$).
For the approximation guarantees the following hold:
	$\alg{FCA}$ returns either an exact $od$-path,
	or an approximate $od$-path via a landmark $\ell_o$
	s.t.
	\(
		D[o,d](t_o)                                    								
		\leq R_o + \overline{\Delta}[\ell_o,d](t_o + R_o)	
		\leq (1+\eps)\cdot D[o,d](t_o) + \psi\cdot R_o			
		\leq (1+\eps+\psi)\cdot D[o,d](t_o)\,,
	\)
where $R_o = D[o,\ell_o](t_o)$ is the minimum-travel-time to the closest landmark, and  $\psi = 1 + \La_{\max}(1+\eps)(1+2\zeta+\La_{\max} \zeta) + (1+\eps)\zeta$ is a cost-metric dependent  constant.
	$\alg{RQA}$ returns, for given recursion budget $r\in \naturals$, an $od$-path that guarantees stretch $1+\s$, where $\s = \s(r) \leq \frac{\eps\cdot (1+\eps/\psi)^{r+1}}{(1+\eps/\psi)^{r+1}-1}$.
\end{theorem}

When $K^*\in \order{n}$ these TD-oracles of \cite{2014-Kontogiannis-Zaroliagis} achieve both \emph{sublinear} query times and \emph{subquadratic} preprocessing requirements. Unfortunately, experimental evidence \cite{2015-Kontogiannis-Michalopoulos-Papastavrou-Paraskevopoulos-Wagner-Zaroliagis} has demonstrated that it may be the case that $K^*\in \OmegaOrder{n}$.

\section{The $\alg{TRAP}$ approximation method}
\label{section:trap-appendix}

In this section, we provide the missing proofs of Section~\ref{section:TRAP-method}.

\begin{proof}[Proof of Lemma~\ref{lemma:upper+lower-approximating-functions-of-TRAP}]
By Assumption~\ref{assumption:Bounded-Travel-Time-Slopes}, for any departure-time $t\in I_k = [t_s = (k-1)\tau , t_f= k\tau)$ from $\ell$ and any destination vertex $v\in V$, the following inequalities hold:
\begin{eqnarray*}
	& &
	-\La_{\min} \leq \frac{D[\ell,v](t) - D[\ell,v](t_s)}{t-t_s} \leq \La_{\max}
	\\
	& \Rightarrow &
	\fbox{$-\La_{\min} (t-t_s) + D[\ell,v](t_s) \leq D[\ell,v](t) \leq D[\ell,v](t_s) + \La_{\max} (t-t_s)$}
	\\[20pt]
	& &
	-\La_{\min} \leq \frac{D[\ell,v](t_f) - D[\ell,v](t)}{t_f-t} \leq \La_{\max}
	\\
	& \Rightarrow &
	\fbox{$\La_{\min} (t_f-t) + D[\ell,v](t_f) \geq D[\ell,v](t) \geq -\La_{\max} (t_f-t) + D[\ell,v](t_f)$}
\end{eqnarray*}
Combining the two inequalities we get the following bounds: $\forall v\in V, \forall t\in I_k$:
\begin{equation*}
	\max	\left\{\begin{array}{c}
						-\La_{\min} t + \La_{\min} t_s + D[\ell,v](t_s)
						\\
                    \La_{\max} t - \La_{\max} t_f + D[\ell,v](t_f)
			\end{array}\right\}
	\leq D[\ell,v](t) \leq
	\min	\left\{\begin{array}{c}
						 \La_{\max} t  - \La_{\max} t_s + D[\ell,v](t_s)
						\\
						-\La_{\min} t + \La_{\min} t_f  + D[\ell,v](t_f)
			\end{array}\right\}
\end{equation*}
Exploiting the fact that each minimum-travel-time function from $\ell$ to any destination $v\in V$ and departure time from $I_k$ respects the above mentioned upper and lower bounds, one could use a simple continuous, pwl approximation of $D[\ell,v]$ within this interval:
\begin{equation*}
\forall t\in I_k,~
	\overline{\de}_k[\ell,v](t) =
	\min	\left\{\begin{array}{c}
						  \La_{\max} t + D[\ell,v](t_s) - \La_{\max} t_s,
						\\
						- \La_{\min} t + D[\ell,v](t_f) + \La_{\min} t_f
		\end{array}\right\}
\end{equation*}
I.e., we consider the lower-envelope of the lines passing via the point $(t_s,D[\ell,v](t_s))$ with the maximum slope $\La_{\max}$, and the point $(t_f,D[\ell,v](t_f))$ with the minimum  slope $-\La_{\min}$.
Analogously, we construct a lower-bounding approximation of $D[\ell,v]$ within $I_k$:
\begin{equation*}
\forall t\in I_k,~
	\underline{\de}_k[\ell,v](t) =
	\max	\left\{\begin{array}{c}
						  \La_{\max} t + D[\ell,v](t_f) - \La_{\max} t_f,
						\\
						- \La_{\min} t + D[\ell,v](t_s) + \La_{\min} t_s
		\end{array}\right\}
\end{equation*}
Figure~\ref{fig:trapezoidal-approximation} shows the (upper and lower) approximations with respect to $D[\ell,v]$ within $[t_s,t_f)$.
\end{proof}

\begin{proof}[Proof of Lemma~\ref{lemma:hole-radius-of-trapezoidal}]
Since $(\underline{t}_m,\underline{D}_m)$ is the intersection of two lines, it is easy to show that:
\begin{eqnarray*}
	\underline{t}_m &=& \frac{D[\ell,v](t_s) - D[\ell,v](t_f)}{\La_{\min} +\La_{\max}} + \frac{\La_{\min} t_s + \La_{\max} t_f}{\La_{\min} + \La_{\max}}
	\\[10pt]
	\underline{D}_m &=& \frac{\La_{\max} D[\ell,v](t_s) + \La_{\min} D[\ell,v](t_f)}{\La_{\min}+\La_{\max}} - \frac{\La_{\min}\cdot\La_{\max}}{\La_{\min}+\La_{\max}}\cdot(t_f-t_s)
\end{eqnarray*}

Analogously, $(\overline{t}_m,\overline{D}_m)$ is also the intersection of two lines. Therefore:
\begin{eqnarray*}
	\overline{t}_m &=& \frac{D[\ell,v](t_f) - D[\ell,v](t_s)}{\La_{\min}+\La_{\max}} + \frac{\La_{\min} t_f + \La_{\max} t_s}{\La_{\min}+\La_{\max}}
	\\[10pt]
	\overline{D}_m &=& \frac{\La_{\max} D[\ell,v](t_f) + \La_{\min} D[\ell,v](t_s)}{\La_{\min}+\La_{\max}} + \frac{\La_{\min}\La_{\max}}{\La_{\min}+\La_{\max}}(t_f-t_s)
\end{eqnarray*}

We start with the upper bound on the maximum absolute error:
\begin{eqnarray*}
	\lefteqn{MAE[\ell,v](I_k)
			\leq \overline{D}_m - \underline{D}_m}
	\\
	&=&		\frac{\La_{\max} D[\ell,v](t_f) + \La_{\min} D[\ell,v](t_s)}{\La_{\min}+\La_{\max}} + \frac{\La_{\min}\La_{\max}}{\La_{\min} + \La_{\max}}(t_f-t_s)
	\\
	&-&		\frac{\La_{\max} D[\ell,v](t_s) + \La_{\min} D[\ell,v](t_f)}{\La_{\min} + \La_{\max}} + \frac{\La_{\min} \La_{\max}}{\La_{\min} + \La_{\max}} (t_f-t_s)
	\\
	&=& 	\frac{(\La_{\max}-\La_{\min}) [ D[\ell,v](t_f) - D[\ell,v](t_s) ]
			+ 2\La_{\min} \La_{max} (t_f-t_s)}{\La_{\min}+\La_{\max}}
	\\
	&=& 	\frac{(\La_{\max}-\La_{\min}) [D[\ell,v](t_f) - D[\ell,v](t_s)]/(t_f-t_s)
			+ 2\La_{\min}\La_{max}}{\La_{\min}+\La_{\max}} (t_f-t_s)
	\\
	&\DueTo{\leq}{As.\ref{assumption:Bounded-Travel-Time-Slopes}}&
			\frac{(\La_{\max}-\La_{\min}) \La_{\max}
			+ 2\La_{\min}\La_{max}}{\La_{\min} + \La_{\max}} (t_f-t_s)
	= \La_{\max} \tau
\end{eqnarray*}

Recall now about $\overline{\delta}_k[\ell,v]$ that: $\forall t\in I_k$,
\begin{eqnarray*}
	\overline{\delta}_k[\ell,v](t)
	&\leq& \underline{\delta}_k[\ell,v](t) + MAE[\ell,v](I_k)
	\leq \underline{\delta}_k[\ell,v](t) + \La_{\max} \tau
	\\
	&\leq&  D[\ell,v](t)\cdot \left( 1 + \frac{ \La_{\max} \tau }{\underline{\delta}_k[\ell,v](t)}\right)
\end{eqnarray*}
Our goal is to assure that this last upper bound of $\overline{\delta}_k[\ell,v](t)$ is in turn upper-bounded by $(1+\eps)\cdot D[\ell,v](t)$. Based on the expression of $\underline{\delta}_k[\ell,v](t)$, and exploiting also the fact that $\tau \geq \max\{ t-t_s, t_f-t\}$, a \emph{sufficient condition} for this to hold, is the following:
\begin{equation*}
	\begin{array}{c}
	D[\ell,v](t_s) \geq \left(\La_{\min} +\frac{\La_{\max}}{\eps}\right)\tau
	~~\OR~~
	D[\ell,v](t_f) \geq \left(\La_{\max} +\frac{\La_{\max}}{\eps}\right)\tau
	\end{array}
\end{equation*}
This sufficient condition is independent of the actual departure time $t\in I_k$, and only depends on the travel-time values at the endpoints $t_s$ and $t_f$, and also on the length $\tau$ of the departure-times subinterval that we choose.
\end{proof}

\section{The $\alg{FLAT}$ and $\alg{TRAPONLY}$ oracles}
\label{section:FLAT+TRAP-detailed-analyses}

We start by providing the missing proof of Lemma~\ref{lemma:period-vs-network-size}.

\begin{proof}[Proof of Lemma~\ref{lemma:period-vs-network-size}]

All the properties mentioned in the statement of the Lemma, are consequences of Assumption~\ref{assumption:Travel-Time-vs-Dijkstra-Rank}.
We also exploit the fact that the free-flow diameter $diam(G,\underline{D})$ of $G$ corresponds to the maximum possible Dijkstra-Rank, which is equal to $n$, assuming that the graph is strongly connected, if we use as root of the Dijkstra tree the origin of a longest minimum-travel-time path in $(G,\underline{D})$.
In particular, we proceed with the explanation of each property separately:

\begin{itemize}

\item[(i)] Recall that $f(n), g(n) \in \polylog(n) = n^{\Order{\frac{\log\log(n)}{\log(n)}}}$ and $\la\in\order{\frac{\log(n)}{\log\log(n)}}$.
	We start by providing the upper bound of $T$:
	\[
		\begin{array}{rrcl}
		& diam(G,\underline{D}) & \leq & g(n)\cdot n^{1/\la}
		\\
		\Rightarrow & T & = & (diam(G,\underline{D}))^{1/\nu} \leq (g(n))^{1/\nu}\cdot n^{1/(\nu\la)}
		\\
		& & = & n^{\frac{1}{\nu\la} + \frac{1}{\nu}\cdot \Order{\frac{\log\log(n)}{\log(n)}}}
					= n^{\frac{1}{\nu\la}\left[1 + \la\cdot \Order{\frac{\log\log(n)}{\log(n)}}\right]}
					= n^{\frac{ 1 + \order{1} }{ \nu\la }}
		\end{array}
	\]
	As for the lower bound of $T$, we have:
	\[
		\begin{array}{rrcl}
		& diam(G,\underline{D}) & \geq & \left(\frac{n}{f(n)}\right)^{1/\la}
		\\
		\Rightarrow & T & 	= & (diam(G,\underline{D}))^{1/\nu} \geq \left(\frac{n}{f(n)}\right)^{1/(\nu\la)}
								= n^{ \frac{1}{\nu\la}\cdot\left[1 - \Order{\frac{\log\log(n)}{\log(n)}}\right]}
								= n^{ \frac{ 1 - \order{1} }{ \nu\la } }
		\end{array}
	\]

\item[(ii)]
	For the upper bound on $F$ we have:
	\[
		\begin{array}{rcl}
		F &=& \max_{\ell\in L} \{|B[\ell,\underline{R}]|\}
		\\
		& \leq & f(n) \underline{R}^{\la} = f(n) T^{\theta\la}
		\leq n^{\frac{\theta}{\nu}\cdot\left[1 +\left(\la + \frac{\nu}{\theta}\right)\cdot\Order{\frac{\log\log(n)}{\log(n)}}\right]}
		= n^{ \frac{ [1+\order{1}] \theta }{\nu} }
	\end{array}
	\]
	where the last step is because we consider instances with $\la\in\order{\frac{\log(n)}{\log\log(n)}}$, and moreover due to the fact that we set our (yet unspecified) tuning parameter $\theta$ so that $\frac{\nu}{\theta}\in\Order{1}$.

	For the lower bound on $F$ we have:
	\[
		\begin{array}{rcl}
			\underline{R} &\leq & g(n) \cdot F^{1/\la}
			\\
			F & \geq & \left(\frac{\underline{R}}{g(n)}\right)^{\la} = \frac{T^{\theta\la}}{(g(n))^{\la}}
			\geq n^{[1-\order{1}]\frac{\theta}{\nu} - {\la}\Order{\frac{\log\log(n)}{\log(n)}}}
			\\
			&=& n^{\frac{\theta}{\nu} \cdot
						\left[ 1 - \left(1+\frac{\nu\la}{\theta}\right)\cdot\Order{\frac{\log\log(n)}{\log(n)}}\right]}
			= n^{\frac{ [ 1 - \order{1}]\theta }{\nu} }
		\end{array}
	\]
	where the last equality is, again, valid since $\la\in\order{\frac{\log(n)}{\log\log(n)}}$ and $\frac{\nu}{\theta}\in\Order{1}$.

\end{itemize}
\end{proof}


\subsection{Analysis of the $\alg{TRAPONLY}$ oracle.}
\label{section:TRAPONLY-detailed-analysis}

Recall that the preprocessing of the $\alg{TRAPONLY}$ oracle is based solely on $\alg{TRAP}$ for computing travel-time summaries, while the query algorithm is an appropriate variant of $\alg{RQA}$ (we call it $\alg{RQA}^+$) which additionally grows a small $\alg{TDD}$ ball as soon as it settles a new landmark, in order to compute ``on the fly'' the exact minimum-travel-times (rather than evaluating preprocessed summaries, which do not exist) towards the nearby destinations from it. Theorem~\ref{thm:TRAPONLY-Complexities} provides the performance of the $\alg{TRAPONLY}$ oracle.

Corollaries \ref{cor:TRAPONLY-Efficient-Complexities_scaling-IMPLIES-stretch} and \ref{cor:TRAPONLY-Efficient-Complexities_stretch-IMPLIES-scaling} explore the conditions under which sublinear query-time and/or subquadratic preprocessing complexities can be guaranteed. We provide here their proofs.

\begin{proof}[Proof of Corollary~\ref{cor:TRAPONLY-Efficient-Complexities_scaling-IMPLIES-stretch}]

We consider scaled TD-instances with $\frac{1}{\nu\la} = \alpha$ for some constant $\alpha \in (0,1)$.
We start with the sublinearity of the query time. For arbitrary constant $\delta\in(0,1)$ such that 
\(
	\delta = \max\left\{\omega\cdot(r+1) , \omega r + \frac{\theta}{\nu} \right\},
\)
if we choose $\omega =  \frac{\theta}{\nu}$ then we have 
\(	
	\delta = \omega\cdot(r+1) 
	\Leftrightarrow 
	\omega = \frac{\theta}{\nu} = \frac{\delta}{r+1} 
	\Leftrightarrow 
	\theta = \frac{\delta\nu}{r+1}
\) 
and the \emph{sublinearity} of the expected query time $n^{\delta}$ is guaranteed.

We continue with the demand for subquadratic preprocessing requirements. Fix now some $\beta\in(0,1)$. In order to assure preprocessing time and space (roughly) $n^{2-\beta + \order{1}}$, it is necessary to assure that 
\[
	\beta \leq \omega - \alpha\cdot(1-\theta) = \frac{\delta}{r+1} - \alpha + \frac{\alpha\delta\nu}{r+1} 
	\Leftrightarrow 
	\fbox{$r\leq \frac{\delta\cdot(1+\alpha\nu)}{\alpha+\beta} - 1$}
\]	
\end{proof}

\begin{remark}
	Note that the query-time performance of $\alg{RQA}^+$ is equal to that of $\alg{RQA}$ (cf. \cite{2014-Kontogiannis-Zaroliagis}), if $\omega = \frac{\theta}{\nu}$. Moreover, $\beta\leq a^2\nu$ implies that $\frac{\delta\cdot(1+\alpha\nu)}{\alpha+\beta}\geq\frac{\delta}{\alpha}$.
\end{remark}

\begin{proof}[Proof of Corollary~\ref{cor:TRAPONLY-Efficient-Complexities_stretch-IMPLIES-scaling}]

We know that, for those instances for which it works, $\alg{TRAPONLY}$ achieves stretch $1+\eps\cdot\frac{\left(1+\eps/\psi\right)^{r+1}}{\left(1+\eps/\psi\right)^{r+1} - 1}$.
We must assure then that
\[
	\frac{\left(1+\eps/\psi\right)^{r+1}}{\left(1+\eps/\psi\right)^{r+1} - 1}
	\leq k
	\Rightarrow
	r\geq \ceil{\frac{\log\left(\frac{k}{k-1}\right)}{\log(1+\eps/\psi)}} - 1 = \eta(k)
\]
Recall that the maximum recursion budget of $\alg{TRAPONLY}$ is $\floor{\frac{\delta\cdot(1 + \alpha\nu)}{\alpha + \beta}} - 1$ (cf Corollary~\ref{cor:TRAPONLY-Efficient-Complexities_scaling-IMPLIES-stretch}). A sufficient condition for guaranteeing the required budget for having stretch $1+k\cdot\eps$ is thus the following:
\[
	\frac{\delta\cdot(1 + \alpha\nu)}{\alpha + \beta} - 2 \geq \eta(k)
	\Rightarrow
	\fbox{$\alpha\leq\frac{\delta - [\eta(k) + 2]\cdot\beta}{\eta(k) + 2 -\delta\nu}$}
\]
\end{proof}

\subsection{Analysis of the $\alg{FLAT}$ oracle.}
\label{section:FLAT-detailed-analysis}

We provide in this section the proof of the corollaries that showcase appropriate parameter-tunings for achieving sublinear query time and subquadratic preprocessing requirements.

\begin{proof}[Proof of Corollary~\ref{cor:FLAT-Parameter-Tuning_scaling-IMPLIES-stretch}]

Again we set $\omega = \frac{\delta}{r+1}$, in order to achieve (sublinear) query time $n^{\delta}$.
We then set
\(
	\frac{2\theta}{\nu} = 1 + \alpha\cdot(1-\theta)
	\Leftrightarrow
	\theta = \frac{1 + \alpha}{2/\nu + \alpha}\,.
\)
Observe that, for this value of $\theta$, it is guaranteed that $\frac{\nu}{\theta} = \frac{2 + \alpha\nu}{1 + \alpha} < 3$, as was assumed in the proof of Lemma~\ref{lemma:period-vs-network-size}.
To guarantee subquadratic preprocessing requirements $n^{2 - \beta + \order{1}}$, we must assure that:
\begin{eqnarray*}
	&& 2-\beta \geq 2 - \omega + \alpha\cdot(1-\theta)
	\Leftrightarrow
	\beta \leq \omega - \alpha\cdot(1-\theta)
	\Leftrightarrow
	\beta \leq \frac{\delta}{r+1} - \alpha\cdot\frac{2/\nu-1}{2/\nu + \alpha}
	\\
	\Leftrightarrow
	&& \fbox{$r \leq \frac{\delta}{\alpha}\cdot\frac{\frac{2}{\nu} + \alpha}{\frac{\beta}{\alpha} \left(\frac{2}{\nu} + \alpha\right) + \left(\frac{2}{\nu}-1\right)} - 1$}
\end{eqnarray*}
The approximation guarantee is, again, the one provided by the $\alg{RQA}$ query algorithm.
\end{proof}

\begin{proof}[Proof of Corollary~\ref{cor:FLAT-Efficient-Complexities_stretch-IMPLIES-scaling}]

Recall that, in order to assure a stretch factor $1+k\cdot\eps$, we must set the recursion budget $r \geq \ceil{\frac{\log\left(\frac{k}{k-1}\right)}{\log\log(1+\eps / \psi)}} - 1 = \eta(k)$, as in the proof of Corollary~\ref{cor:TRAPONLY-Efficient-Complexities_stretch-IMPLIES-scaling}.

From Theorem~\ref{thm:FLAT-Complexities} we also have an upper bound on the recursion budget. Thus, as $\beta\downarrow 0$, a sufficient condition for the recursion budget, so that the required stretch is achieved is the following:
\[
	\eta(k) + 2 
	\leq \frac{\delta}{\alpha} \cdot \frac{2/\nu + \alpha}{2/\nu - 1}
	\Rightarrow
	\fbox{$\alpha \leq \frac{2\delta}{[\eta(k)+2]\cdot(2-\nu) - \delta\nu}$}
\]
It is now straightforward that the expected query-time is indeed $n^{\delta+\order{1}}$, whereas the preprocessing requirements are $n^{2-\order{1}}$ since we consider a very small value for $\beta$.
\end{proof}

\section{Analysis of the $\alg{HORN}$ oracle}
\label{section:HORN-detailed-analysis}


%
The construction of the travel-time summaries for $\alg{HORN}$ is based on
the $\alg{FLAT}$ ($\alg{BIS+TRAP}$) preprocessing scenario. The queries are
served by the $\alg{HQA}$ query algorithm.
The oracle exploits two fundamental properties:
\begin{itemize}

\item[(i)] the approximation guarantee of a path via some landmark $\ell$ strongly depends on the relative distance of the landmark from the origin $o$, compared to the distance of the destination $d$ from $o$;

\item[(ii)] given that the expected distance of a level-$i$ from the origin is roughly $\frac{1}{\rho_i}$, it is rather unlikely that the first level-$i$ landmark will appear too early or too late (i.e., outside a sufficiently wide ring-stripe around $(o,t_o)$). 

\end{itemize}

Property (i) 
is exploited by the (ESC) criterion in order to handle the exceptional case where a higher-level landmark (which also
happens to be informed) appears before the first informed landmark from the appropriate level.
Property (ii) is actually an event that holds with high probability, as is shown in the detailed analysis of the oracle, and is exploited by the (ALH) criterion. Therefore, the event that an informed landmark appears which is also at the right distance, whereas the previously discovered landmarks (most likely of smaller levels) were uninformed, reveals an asymptotic bound for the unknown Dijkstra-Rank $\G[o,d](t_o)$ of the destination.

We now provide a sequence of lemmata which will eventually be used in the proof our main technical result, concerning the complexities of $\alg{HORN}$ mentioned in Theorem~\ref{thm:HORN-Complexities} (cf. Section~\ref{section:HORN-oracle}).

We start with an upper bound on the free-flow distance of a discovered landmark $\ell_o$ from $d$. At this point we do not require that $d\in C[\ell_o]$.
\begin{lemma}
\label{lemma:upper-bound-on-free-flow-distance}
	Let $\ell_o\in L$ be a landmark discovered by $\alg{HQA}$. Then it holds that
\(
	\underline{D}[\ell_o,d] \leq \frac{\zeta}{1-\La_{\min}}\cdot D[o,\ell_o](t_o) + D[o,d](t_o)\,.
\)
\end{lemma}

\begin{proof}[Proof of Lemma~\ref{lemma:upper-bound-on-free-flow-distance}]

	By Assumption~\ref{assumption:Bounded-Opposite-Trips} we know that:
\begin{equation}
	\label{eq:revert-travel-time-from-landmark-to-origin}
	D[\ell_o,o](t_o) \leq \zeta\cdot D[o,\ell_o](t_o)
\end{equation}
By Assumption~\ref{assumption:Bounded-Travel-Time-Slopes} we also know that:
\begin{equation}
	\label{eq:upper-bound-travel-time-from-landmark-to-origin}
	\begin{array}{rrl}
	\forall x>0, & - \La_{\min}\cdot x \leq & D[\ell_o,o](t_o) - D[\ell_o,o](t_o - x)
	\\
	\Rightarrow & D[\ell_o,o](t_o - x) \leq & D[\ell_o,o](t_o) + \La_{\min}\cdot x
	\end{array}
\end{equation}
We look for a particular departure-time $t_o - x_o$, and the corresponding minimum-travel-time $D[\ell_o,o](t_o-x_o)$, so as to be at the origin $o$ exactly at time $t_o$. That is:
\begin{eqnarray}
	\nonumber
	t_o &=& t_o-x_o + D[\ell_o,o](t_o-x_o)
	\\
	\nonumber
	\Rightarrow~
	x_o &=& D[\ell_o,o](t_o-x_o)
	\DueTo{\leq}{(\ref{eq:upper-bound-travel-time-from-landmark-to-origin})} D[\ell_o,o](t_o) + \La_{\min}\cdot x_o
	\\
	\label{eq:upper-bound-on-landmark-to-origin-distance}
	\Rightarrow~
	x_o &=& \leq \frac{D[\ell_o,o](t_o)}{1-\La_{\min}}
	\DueTo{\leq}{(\ref{eq:revert-travel-time-from-landmark-to-origin})}
		\frac{\zeta}{1-\La_{\min}}\cdot D[o,\ell_o](t_o)
\end{eqnarray}
Finally, we upper-bound the free-flow distance of $\ell_o$ from $d$ by exploiting the triangle inequality:
\begin{eqnarray*}
	\underline{D}[\ell_o,d]
	&\leq& D[\ell_o,d](t_o-x_o) \leq D[\ell_o,o](t_o-x_o) + D[o,d](t_o)
	\\
	&\DueTo{\leq}{(\ref{eq:upper-bound-on-landmark-to-origin-distance})}&
	\frac{\zeta}{1-\La_{\min}}\cdot D[o,\ell_o](t_o) + D[o,d](t_o)
\end{eqnarray*}
which is exactly the desired inequality.
\end{proof}

The following lemma provides an upper bound on the approximation guarantee of $\alg{HQA}$, when an $(ESC)$-termination occurs.
\begin{lemma}
	\label{lemma:early-stopping-criterion}
	Assume that $\ell_o\in L$ is informed (i.e., $d\in C[\ell_o]$) and settled by the initial $\alg{TDD}$ ball that $\alg{HQA}$ grows from $(o,t_o)$. Then, for any given value $\varphi\geq 1$, if $\alg{HQA}$ terminated due to occurrence of (ESC), the reported travel-time is at most a $\left(1+\eps+\frac{\psi}{\varphi\cdot(r+1)}\right)$-approximation of the minimum travel-time $R_d = D[o,d](t_o)$.
\end{lemma}
%
\begin{proof}[Proof of Lemma~\ref{lemma:early-stopping-criterion}]
	Let $D[o,\ell_o](t_o) = R_o$. Then, $R_d \leq \overline{\De}[o,d](t_o) = R_o + \overline{\De}[\ell_o,d](t_o+R_o)$.
By Theorem~\ref{thm:BIS+RQA-Complexities} (cf.~also the analysis of $\alg{FCA}$ in \cite{2014-Kontogiannis-Zaroliagis})
we can easily deduce that:
\begin{eqnarray*}
	&& \frac{\overline{\De}[\ell_o,d](t_o+R_o)}{R_o} \leq (1+\eps)\frac{R_d}{R_o} + \psi - 1
	\\
	&\DueTo{\Rightarrow}{\mbox{\tiny $(ESC)$-termination}}
	& (1+\eps)\cdot\varphi\cdot(r+1) + \psi - 1 \leq \frac{\overline{\Delta}[\ell_o,d](t_o+R_o)}{R_o} \leq (1+\eps)\frac{R_d}{R_o} + \psi - 1
	\\
	& \Rightarrow
	& \frac{R_d}{R_o} \geq \varphi\cdot(r+1)
	\\
	& \Rightarrow
	& 1 + \eps + \frac{\psi R_o}{R_d} \leq 1+\eps + \frac{\psi}{\varphi\cdot(r+1)}
\end{eqnarray*}
Since 	$1 + \eps + \frac{\psi R_o}{R_d}$ is an upper bound on the approximation guarantee provided by $\alg{FCA}$ (cf. Theorem~\ref{thm:BIS+RQA-Complexities}), which is indeed simulated by $\alg{HQA}$ until the determination of the appropriate level in the hierarchy, we conclude that the eventual solution that will be provided by $\alg{HQA}$ is at least as good, since the $(ESC)$-termination returns the best approximate solution seen so far via an informed landmark, among which is also the one that goes via $\ell_o$.
\end{proof}

We proceed now by studying the first appearance of a level-$i$ landmark within the unique outgoing ball from $(o,t_o)$. The next lemma shows that, with high probability, this first appearance of a level-$i$ landmark will take place in the following \term{ring} for level-$i$:
\begin{eqnarray*}
	RING[o;i](t_o)
	&:=& B\left[o~;~(N_i)^{\delta / (r+1) }\cdot \ln(n)\right]\left(t_o\right)
	~ \setminus ~ B\left[o~;~\frac{(N_i)^{\delta / (r+1) }}{\ln(n)}\right]\left(t_o\right)
	\\
	&=& B\left[o~;~\frac{\ln(n)}{\rho_i}\right]\left(t_o\right)
	~ \setminus ~ B\left[o~;~\frac{1}{\rho_i\ln(n)}\right]\left(t_o\right)
\end{eqnarray*}
since $N_i = n^{(\g^i-1)/\g^i}$ and
$\rho_i = n^{-\delta\cdot(\g^i-1)/[(r+1)\g^i]} = N_i^{-\delta/(r+1)}$.
\begin{lemma}
\label{lemma:probability-of-landmark-appearing-in-ring}
$\forall i\in[k]$, there is at least one level-$i$ landmark in $RING[o;i](t_o)$, with probability $1 - \Order{\frac{1}{n}}$.
\end{lemma}
\begin{proof}[Proof of Lemma~\ref{lemma:probability-of-landmark-appearing-in-ring}]

Consider any subset of vertices $S\subseteq V$, of size $s = |S|\in \naturals$. The probability that none of the vertices in $S$ is a level-$i$ landmark (i.e., from $L_i$) is $(1-\rho_i)^s \leq \exp(-s\rho_i)$.

Observe now that, for $i\in[k]$, $s_i = |RING[o;i](t_o)| = \frac{\ln(n)}{\rho_i} - \frac{1}{\rho_i\ln(n)}$. Thus, we conclude that:
\(
	\Prob{ |RING[o;i](t_o) \intersection L_i| = 0 }
	\leq \exp(-s_i\cdot\rho_i) = \frac{\exp\left(\frac{1}{\ln(n)}\right)}{n} \in \Order{\frac{1}{n}}\,.
\)
\end{proof}
The next lemma states that, given the actual Dijkstra-Rank $\Gamma[o,d](t_o)$ that we seek for, the level-$i$ landmark (if any) that is settled within $RING[o;i](t_o)$ is indeed informed about $d$, for all $i$ such that $N_i \geq \Gamma[o,d](t_o)$.
\begin{lemma}
	\label{lemma:HQA-Success-In-Appropriate-Level}
	For $i\in [k]$, let $\Gamma[o,d](t_o)\leq N_i = n^\frac{\g^{i}-1}{\g^{i}}$. Assuming that $(ALH)$-termination occurred, the first level-$i$ landmark $\ell_{i,o}\in L_i\intersection RING[o;i](t_o)$ that is settled by the initial $\alg{TDD}$ ball grown by $\alg{HQA}$, has $d\in C[\ell_{i,o}]$.
\end{lemma}
%
\begin{proof}[Proof of Lemma~\ref{lemma:HQA-Success-In-Appropriate-Level}]

Let $R_{i,o} = D[o,\ell_{i,o}](t_o)$ and recall that $R_d = D[o,d](t_o)$. Assume also that $R_d > R_{i,o}$, because otherwise an exact solution will be anyway discovered before $\ell_{i,o}$ is settled. Then we have, by Lemma~\ref{lemma:upper-bound-on-free-flow-distance}:
\begin{eqnarray*}
	\underline{D}[\ell_{i,o},d]
	& \leq &	\frac{\zeta}{1-\La_{\min}} R_{i,o} + R_d
	\\
	& < &		\left(1+\frac{\zeta}{1-\La_{\min}}\right)\cdot R_d
	\\
	&\DueTo{\leq}{\mbox{\tiny As.}~\ref{assumption:Travel-Time-vs-Dijkstra-Rank}}&
	\left(1+\frac{\zeta}{1-\La_{\min}}\right)\cdot g(n)\cdot n^{(\g^i-1) / (\la \g^i)}
	\\
	\Rightarrow
	\underline{\Gamma}[\ell_{i,o},d]
	& \DueTo{\leq}{\mbox{\tiny As.}~\ref{assumption:Travel-Time-vs-Dijkstra-Rank}}&  f(n)\cdot (\underline{D}[\ell_{i,o},d])^{\la}
	\leq f(n)\cdot g^{\la}(n)\cdot\left(1 + \frac{\zeta}{1-\La_{\min}}\right)^{\la} \cdot n^{(\g^i-1) / \g^i}
	\\
	&=& 	n^{(1+\la)\Order{\frac{\log\log(n) }{ \log(n) }}}
			\cdot n^{\la\log\left(1+\frac{\zeta}{1-\La_{\min}}\right) / \log(n)}
			\cdot n^{(\g^i-1) / \g^i}
	\\
	&=& 	n^{(\g^i-1) / \g^i + \order{1}}
	\leq	n^{(\g^i-1) / \g^i + \xi_i}
\end{eqnarray*}
for any
	\(
		\xi_i
		\geq \frac{(1+\la)\cdot \log\log(n) + \la\log\left(1 + \frac{\zeta}{1-\La_{\min}}\right) }{\log(n)}
	\)
which is certainly true since $\la\in\order{\frac{\log(n)}{\log\log(n)}}$ and $f(n),g(n)\in \polylog(n) = n^{\Order{\frac{\log\log(n)}{\log(n)}}}$. The last inequality implies that $d\in C[\ell_{i,o}]$.
\end{proof}
The next lemma states that at no level-$j$ in the hierarchy earlier than the \term{appropriate level} $i^*$ corresponding to $\Gamma[o,d](t_o)$ (cf. Theorem~\ref{thm:HORN-Complexities}), may $\alg{HQA}$ find a landmark which contains $d$ in its coverage (and which would then cause an incorrect ``guess'' of the appropriate level for the query at hand), provided that no $(ESC)$-termination occurred.
\begin{lemma}
	\label{lemma:HQA-Failure-In-Previous-Levels}
	For $i\geq 1$, assume that $\Gamma[o,d](t_o) > N_i = n^{(\g^i-1) / \g^i}$ and, while growing the ball from $(o,t_o)$, no (ESC) occurs. Then, $\forall 1\leq j\leq i$ no level-$j$ landmark $\ell$ in $RING[o;j](t_o)$ contains $d$ in its coverage: $d\notin \union_{1\leq j \leq i} \union_{\ell\in L_j\intersection RING[o;j](t_o)} C[\ell]$.
\end{lemma}
%
\begin{proof}[Proof of Lemma~\ref{lemma:HQA-Failure-In-Previous-Levels}]

Recall again that $R_{j,o} = D[o,\ell_{j,o}](t_o)$ is the distance from the origin $o$ to the first level-$j$ landmark that we meet in $RING[o;j](t_o)$.
We start by providing a simple proof for the cases of $1\leq j\leq i-1$. We shall then handle the case $j=i$ separately, since it is a little bit more involved. So, fix arbitrary $j\in [i-1]$. For any $\ell_{j,o}\in RING[o;j](t_o)\intersection L_j$ it holds, by the triangle inequality, that:
\begin{eqnarray*}
	\lefteqn{D[\ell_{j,o},d](t_o+R_{j,o})}
	\\
	&\geq& D[o,d](t_o) - D[o,\ell_{j,o}](t_o)
	\\
	&\DueTo{\geq}{\mbox{\tiny As.}~\ref{assumption:Travel-Time-vs-Dijkstra-Rank}}&
	\frac{1}{f^{1/\la}(n)} (\Gamma[o,d](t_o))^{1/\la} - g(n) (\G[o,\ell_{j,o}](t_o))^{1/\la}
	\\
	&>&
	\frac{1}{f^{1/\la}(n)} n^{ (\g^i-1) / (\la\g^i) } - g(n) \ln^{1/\la}(n) n^{ a (\g^j-1) / ((r+1)\la\g^j) }
	\\
	&=&
	n^{ (\g^j-1) / (\la\g^j) }\cdot
	\left[
		\frac{n^{ (\g^{i-j} - 1) / (\la\g^i) }}{f^{1/\la}(n)}
		-g(n) \frac{\ln^{1/\la}(n)}{ n^{ \left(1 - \frac{a}{r+1}\right) (\g^j-1) / (\la\g^j) }}
	\right]
\end{eqnarray*}
\begin{eqnarray*}
	\lefteqn{
		\Rightarrow
		\Gamma[\ell_{j,o},d](t_o+R_{j,o})
		\DueTo{\geq}{\mbox{\tiny As.}~\ref{assumption:Travel-Time-vs-Dijkstra-Rank}}
	\left(\frac{D[\ell_{j,o},d](t_o+R_{j,o})}{g(n)}\right)^{\la}}
	\\
	&>&
	n^{ (\g^j-1) / (\g^j) }\cdot
	\left[
		\frac{n^{ (\g^{i-j} - 1) / (\la\g^i) }}{g^{\la}(n) f^{1/\la}(n)}
		- \frac{\ln^{1/\la}(n)}{g^{\la-1}(n) n^{ \left(1 - \frac{\delta}{r+1}\right) (\g^j-1) / (\la\g^j) }}
	\right]^{\la}
\end{eqnarray*}
Observe now that, from this last inequality and Assumption~\ref{assumption:growth-of-free-flow-ball-sizes}, the following is deduced:
\begin{eqnarray*}
	\underline{\Gamma}[\ell_{j,o},d]
	&>&
	\frac{n^{ (\g^j-1) / (\g^j) }}{\polylog(n)}\cdot
	\left[
		\frac{ n^{ (\g^{i-j} - 1) / (\la\g^i) }}{ g^{\la}(n) f^{1/\la}(n) }
		- \frac{ \ln^{1/\la}(n) }{ g^{\la-1}(n) n^{ \left(1 - \frac{\delta}{r+1}\right) (\g^j-1) / (\la\g^j) } }
	\right]^{\la}
	\\
	&\geq& 	n^{ \frac{\g^j-1 }{ \g^j } + \frac{\g^{i-j}-1}{ \g^i } - \Order{\la^2\cdot\frac{\log\log(n) }{ \log(n) }} }
	\\
	&>& n^{(\g^j-1)/(\g^j) + \xi_j}
\end{eqnarray*}
The last inequality holds when $\la\in\order{\sqrt{\frac{\log(n)}{\log\log(n)}}}$, and $\xi_j < \g^{-j} - \g^{-i} = \frac{\g-1}{\g^{j+1}}$. The last inequality implies also that $d\notin C[\ell_{j,o}]$.

We shall now study separately the case $j=i$. Apart from the containment of the level-$i$ landmarks in $RING[o;i](t_o)$, we also exploit the fact that the first of these landmarks met by the unique $\alg{TDD}$ ball from $(o,t_o)$ does not cause an $(ESC)$-termination  of $\alg{HQA}$. This implies that:
\begin{eqnarray}
    \nonumber
	\frac{\overline{\De}[\ell_{i,o},d](t_o+R_{i,o})}{R_{i,o}}
	&<& (1+\eps)\beta(r+1) + \psi - 1 =: \chi
    \\
    \nonumber	
    \Rightarrow
	R_{i,o}
	&>& \frac{\overline{\De}[\ell_{i,o},d](t_o+R_{i,o})}{\chi}
			\geq \frac{D[\ell_{i,o},d](t_o+R_{i,o})}{\chi}
    \\
    \nonumber
	&\DueTo{\geq}{\mbox{\tiny triangle~ineq.}}& \frac{R_d - R_{i,o}}{\chi}
	\\
    \nonumber
	\Rightarrow
	R_{i,o} &>& \frac{R_d}{1+\chi}
	\\
    \label{ineq:R-i-1}
	&\DueTo{>}{\mbox{\tiny As.}~\ref{assumption:Travel-Time-vs-Dijkstra-Rank}}&
	\frac{ n^{ (\g^i-1) / (\la \g^i) } }{ (1+\chi) f^{1/\la}(n) }
\end{eqnarray}
Nevertheless, we also know that for any $\ell_{i,o}\in RING[o;i](t_o)$ the following holds:
\begin{eqnarray}
    \nonumber
	\Gamma[o,\ell_{i,o}](t_o)
	&\leq&  n^{\delta(\g^i-1) / ((r+1) \g^i)} \cdot \ln(n)
	\\
    \nonumber
	\DueTo{\Rightarrow}{\mbox{\tiny As.}~\ref{assumption:Travel-Time-vs-Dijkstra-Rank}}
	R_{i,o} 	= D[o,\ell_{i,o}](t_o)
				&\leq& g(n)\left(\G[o,\ell_{i,o}](t_o)\right)^{1/\la}
    \\
	 \nonumber
				&\leq &  g(n) n^{\delta(\g^i-1) / ((r+1)\la\g^i)} \ln^{1/\la}(n)
    \\
	\label{ineq:R-i-2}
             & \leq & \frac{ n^{ (\g^i-1) / (\la \g^i) } }{ (1+\chi) f^{1/\la}(n) }~
\end{eqnarray}
Inequality (\ref{ineq:R-i-2}) holds if and only if
\begin{eqnarray*}
    n^{ \left(1 - \frac{\delta}{r+1}\right)(\g^i-1) / (\la \g^i) }
	&\geq& (1+\chi) g(n) f^{1/\la}(n) \ln^{1/\la}(n)
	\\
	\Leftrightarrow
	\left(1 - \frac{\delta}{r+1}\right)\frac{\g^i-1}{\g^i}\log(n) - \log(\ln(n)) - \log(f(n))
	&\geq& \la\log(1+\chi) + \la\log(g(n))
	\\
	\Leftrightarrow
	\left(1 - \frac{\delta}{r+1}\right)\frac{\g^i-1}{\g^i}\log(n) - \Order{\log\log(n)}
	&\geq& \la\left[\log(1+\chi) + \Order{\log\log(n)}\right]
\end{eqnarray*}
which is certainly true for $\la \in \order{\frac{\log(n)}{\log\log(n)}}$ and $f(n),g(n)\in \polylog(n)$.
From (\ref{ineq:R-i-1}) and (\ref{ineq:R-i-2}) we are led to a contradiction. Therefore, any level-$i$ landmark in $RING[o;i](t_o)$ either does not possess a travel-time summary for $d$, or causes an early-stopping of $\alg{HQA}$ upon its settlement.
\end{proof}
The last lemma proves that, given that an $(ALH)$-termination occurs, $\alg{HQA}$ achieves the desired approximation guarantee.
\begin{lemma}
	\label{lemma:HQA+ALU-Approximation-Guarantee}
	For $i^*\in \{2,3,\ldots,k+1\}$, let $N_{i^*-1}< \Gamma[o,d](t_o)\leq N_{i^*}$. Assume also that an $(ALH)$-termination occurred. Then, an $\left(1+\eps\frac{(1+\eps/\psi)^{r+1}}{(1+\eps/\psi)^{r+1} - 1}\right)-$approximate solution is returned by $\alg{HQA}$, with probability $1-\Order{\frac{1}{n}}$. The expectation of the query-time is $\Order{(N_{i^*})^a \polylog(n)}$.
\end{lemma}
%
\begin{proof}[Proof of Lemma~\ref{lemma:HQA+ALU-Approximation-Guarantee}]

As it was already explained by Lemmata~\ref{lemma:probability-of-landmark-appearing-in-ring}, \ref{lemma:HQA-Success-In-Appropriate-Level} and \ref{lemma:HQA-Failure-In-Previous-Levels}, we know that a successful ``guess'' of $i^*$ will occur with probability $1-\Order{\frac{1}{n}}$. Then, $\alg{HQA}$ proceeds with the execution of $\alg{RQA}_{i^*}$, whose expected time is $\Order{(N_{i^*})^a \polylog(n)}$ (cf. Theorem~\ref{thm:BIS+RQA-Complexities}).

As for the approximation guarantee, since the analysis for $\alg{RQA}_{i^*}$ is based solely on the quality of the paths via landmarks discovered from ball centers located along a minimum-travel-time $od$-path $p^*\in SP[o,d](t_o)$ (cf. corresponding proof for $\alg{RQA}$ in \cite{2014-Kontogiannis-Zaroliagis}), it suffices to prove that all the level-$(\geq i^*)$ landmarks discovered from ball centers which reside at the (unknown) shortest $od$-path, are informed about $d$.

For an arbitrary ball center $w_{i^*,j}\in p^*$ and its closest level-$i$ landmark for some $i\geq i^*$, $\ell_{i,j}\in L_i$, let $R_{i^*,j} = D[w_{i^*,j},\ell_{i,j}](t_o+D[o,w_{i^*,j}](t_o))$. Then, either $R_d\leq R_{i^*,o}+\ldots+R_{i^*,j}$, in which case an exact solution is returned, or else the following holds:
\begin{eqnarray*}
	\underline{D}[\ell_{i,j},d]
	&\leq& \frac{\zeta}{1-\La_{\min}} R_{i^*,j} + R_d
	< \left(1 + \frac{\zeta}{1-\La_{\min}}\right)\cdot R_d
	\\
	&\DueTo{\leq}{\mbox{\tiny As.}~\ref{assumption:Travel-Time-vs-Dijkstra-Rank}}&
	\left(1 + \frac{\zeta}{1-\La_{\min}}\right)\cdot g(n)\cdot n^{(\g^i-1)/(\la \g^i)}
	\\
	\Rightarrow
	\underline{\Gamma}[\ell_{i,j},d]	
	&\DueTo{\leq}{\mbox{\tiny As.}~\ref{assumption:Travel-Time-vs-Dijkstra-Rank}}&
	f(n)\cdot(\underline{D}[\ell_{i,j},d])^{\la}
	\leq f(n)\cdot \left(1 + \frac{\zeta}{1-\La_{\min}}\right)^{\la}\cdot g^{\la}(n)\cdot n^{(\g^i-1)/\g^i}
	\\
	&\leq& n^{(\g^i-1)/\g^i + \xi_i}
\end{eqnarray*}
for
\(
	\la\log(g(n)) + \log(f(n)) \leq \xi_i\log(n)
	\Leftrightarrow
	\xi_i \geq (\la+1)\cdot\Order{\frac{\log\log(n)}{\log(n)}}
\). Again, a sufficient condition for this is that $\la\in\order{\frac{\log(n)}{\log\log(n)}}$. We therefore conclude that $d\in C[\ell_{i,j}]$.

Since all the discovered landmarks from ball centers along the unknown shortest path are informed, the analysis of $\alg{RQA}_{i^*}$ provides the claimed approximation guarantee.
\end{proof}
Having proved all the required building blocks, we are now ready to provide the proof of our main argument for $\alg{HORN}$.
%
\begin{proof}[Proof of Theorem~\ref{thm:HORN-Complexities}]

We start with a short sketch of the proof, before delving into the details: The claimed preprocessing space and time are proved, mainly due to the appropriate choice of the $\xi_i$ parameters. We thus focus on the more involved analysis of the time-complexity of $\alg{HQA}$. First of all, exploiting Assumptions~\ref{assumption:Bounded-Travel-Time-Slopes} and \ref{assumption:Bounded-Opposite-Trips}, we prove an upper bound on the free-flow travel-time of the destination $d$ from an arbitrary landmark $\ell$, as a function of the (unknown) travel-time $D[o,d](t_o)$. We then prove that if the $(ESC)$-termination criterion occurs, then we already have a satisfactory approximation guarantee. Otherwise, assuming that the Dijkstra Rank $\Gamma[o,d](t_o)$ lies within the two consecutive values $N_{i^*-1}$ and $N_{i^*}$, we exploit the triangle inequality along with Assumptions ~\ref{assumption:growth-of-free-flow-ball-sizes} and \ref{assumption:Travel-Time-vs-Dijkstra-Rank}, to prove that, with high probability, the closest level-$i^*$ landmark $\ell_{i^*,o}$ to the origin $o$ lies in a certain ring (i.e., difference of two cocentric balls from $o$) and is indeed ``\emph{informed}''. Additionally, all the landmarks discovered prematurely with respect to the ring of the corresponding level to which they belong, cannot be ``\emph{informed}'', given the failure of the $(ESC)$-termination criterion. Moreover, discovered landmarks of higher levers, despite being informed, do not lie in the ring of the corresponding levels and therefore do not interrupt the guessing procedure. This then implies that the appropriate variant of $\alg{RQA}$ will be successfully executed and will return a travel-time estimation with the required approximation guarantee and query-time complexity sublinear in $N_{i^*}$ (the targeted Dijkstra-Rank).

We now proceed with the detailed presentation of the proof. For each level $i\in[k]$, the targeted Dijkstra-Rank is $N_i = n^{(\g-1) / \g^i}$. The coverage of each level-$i$ landmark contains $c_i = N_i\cdot n^{\xi_i} = n^{(\g-1) / \g^i + \xi_i}$ destinations, $F_i = (c_i)^{\theta_i / \nu}$ of which are discovered by $\alg{BIS}$ and the remaining $c_i - F_i$ destinations are discovered by $\alg{TRAP}$.

We start with the analysis of the approximation guarantee and the query-time complexity for $\alg{HQA}$.
Observe that the query-time is dominated by the scenario in which there is an $(ALH)$-termination (the discovery of the destination, or an $(ESC)$-termination will only improve the performance of the algorithm).
Assume that for some $i^*\in \{2,3,\ldots,k+1\}$ (which we call the \term{appropriate level} for the query at hand), $N_{i^*-1} = n^\frac{\g^{i^*-1}-1}{\g^{i^*-1}} < \Gamma[o,d](t_o)\leq n^\frac{\g^{i^*}-1}{\g^{i^*}} = N_{i^*}$.
Lemma~\ref{lemma:early-stopping-criterion} assures that, when an $(ESC)$-termination occurs, an $\left(1+\eps+\frac{\psi}{\varphi\cdot(r+1)}\right)$-approximate solution was anyway discovered and thus there is no need to guarantee the success of the guessing procedure. By setting $\varphi = \frac{\eps (r+1)}{\psi (1+\eps/\psi)^{r+1} - 1}$, the approximation guarantee when $(ESC)$-termination occurs is equal to $1+\eps\frac{(1+\eps/\psi)^{r+1}}{(1+\eps/\psi)^{r+1} -1 }$.
Lemma~\ref{lemma:probability-of-landmark-appearing-in-ring} proves that, with probability $1-\Order{\frac{1}{n}}$, at least one appearance of a level-$i^*$ landmark occurs in the appropriate ring: $L_{i^*}\intersection RING[o;i^*](t_o) \neq \emptyset$.
Lemma~\ref{lemma:HQA-Success-In-Appropriate-Level} then proves that, given that an $(ALH)$-termination occurs, the first level-$i^*$ landmark $\ell_{i^*,o}$ settled by the initial $\alg{TDD}$ ball within $RING[o;i^*](t_o)$ is indeed an ``informed'' landmark: $d\in C[\ell_{i^*,o}]$.
Lemma~\ref{lemma:HQA-Failure-In-Previous-Levels} proves that, given an $(ALH)$-termination, no landmark of a previous level $j<i^*$ that was settled before $\ell_{i^*,o}$ may be informed: $\forall \ell\in\union_{j\in[i^*-1]} (L_j\intersection RING[o;j](t_o)),~ d\notin C[\ell]$.
As for landmarks $\ell\in \union_{j=i^*+1}^{k+1} L_j$, if we settle $\ell$ before $\ell_{i^*,o}$ and it happens that $d\in C[\ell]$ (i.e., $\ell$ is ``informed''), this event will not interrupt the guessing procedure of $\alg{HQA}$, because it is not in the ring of the corresponding level: $(B[o](t_o)\intersection\union_{j=1}^{i^*} RING[o;j](t_o) )\intersection (\union_{j=i^*+1}^{k+1} RING[o;j](t_o)) = \emptyset$.
Therefore we conclude that, in the case of an $(ALH)$-termination, with probability $1-\Order{\frac{1}{n}}$, the ``guess'' of the appropriate level-$i^*$ is indeed correct.
Finally, Lemma~\ref{lemma:HQA+ALU-Approximation-Guarantee} demonstrates that, when an $(ALH)$-termination occurs with a successful ``guess'' of $i^*$, a $\left(1+\eps\frac{(1+\eps/\psi)^{r+1}}{(1+\eps/\psi)^{r+1} -1 }\right)$-approximate solution is returned by $\alg{HQA}$.

As for the expected query-time of $\alg{HQA}$, we wish to assure that it is sublinear in the Dijkstra-Rank of the appropriate level $i^*$ of the hierarchy. Suppose that the exponent of sublinearity is $\delta\in(a,1)$, i.e., the query-time is comparable to $\Order{(N_{i^*})^{\delta}}$, where $a\in (0,1)$ is the exponent relating the period of the metric with the network size ($T = n^a$). We focus on the highly-probable event that $\alg{HQA}$, if it terminates due to the (ALH) criterion, makes a correct guess of $i^*$. Conditioned on this event, the expected cost of $\alg{RQA}_{i^*}$ gives the appropriate value for the exponent $\omega_i$ of the landmark-sampling probability:
\[
	\left(\frac{1}{\rho_{i^*}}\right)^{r+1}\log\left(\frac{1}{\rho_{i^*}}\right)\polylog(n)
	= n^{\omega_i(r+1) + \order{1}}
	= (N_{i^*})^{\delta}
	\Rightarrow
	\fbox{$\omega_i = \frac{\delta}{r+1}\frac{\g^{i^*}-1}{\g^{i^*}}$}
\]
The contribution to the expected query-time of $\alg{HQA}$ of the unlikely event that the algorithm makes a wrong guess about the appropriate level of the query at hand is negligible, due to the quite small probability of this happening.

We proceed next with the study of the required time and space for the $\alg{BIS+TRAP}$-based preprocessing of the $\alg{HORN}$ oracle. We wish to bound the preprocessing requirements with $(k+1)\cdot n^{2-\beta}$, for a given $\beta>0$. We shall make the appropriate choices of our tuning parameters so that in each level it holds that the preprocessing requirements are $S_i ~,~ P_i \leq n^{2-\beta}$.

We begin with the determination of the requirements of level-$i$, for each $i\in[k]$. For this level we generate $|L_i| = n^{1-\omega_i}$ landmarks, each of which possesses travel-time summaries for all the $c_i$ destinations contained in its own coverage. Based on the analysis of the preprocessing analysis of $\alg{FLAT}$ (cf. Theorem~\ref{thm:FLAT-Complexities}), but now restricting ourselves within the coverage of each landmark, we know that the overall preprocessing requirements of level-$i$ are bounded as follows:
\[
	\begin{array}{rcl}
	S_i ~,~ P_i
	&\in& \Order{|L_i|\cdot\left( F_i^2\polylog(F_i) + T^{1-\theta_i} c_i\right)}
	\\
	&=& \Order{n^{1-\omega_i}
					\cdot\left[
						n^{\frac{2\theta_i}{\nu}\left(\frac{\g^i-1}{\g^i} +\xi_i\right) + \order{1}}
						+ n^{a(1-\theta_i) + \frac{\g^i-1}{\g^i} + \xi_i}
					\right]}
	\\
	&\subseteq& \Order{n^{1-\omega_i}
					\cdot\left[
						n^{\frac{2\theta_i}{\nu} + \order{1}}
						+ n^{a(1-\theta_i) + 1}
					\right]}
	\end{array}
\]
provided that $\xi_i\leq \g^{-i}$. We choose again $\theta_i = \frac{1+a}{2/\nu+a}$, which in turn assures that
\[
	S_i ~,~ P_i \in n^{2-\omega_i + a(1-\theta_i) + \order{1}}
\]

We correlate our demand for subquadratic preprocessing with the recursion budget (and thus, the approximation guarantee) that we can achieve:
\[
	\begin{array}{rcl}
	&& 2-\omega_i + a\cdot(1-\theta_i)
	< 2 - \beta < 2
	\\
	\Leftrightarrow
	&& \beta < \omega_i - a \cdot (1-\theta_i)
	= \frac{\delta(1-\g^{-i})}{r+1} - a\cdot \frac{2 / \nu - 1}{2 / \nu + a}
	\\
	\Leftrightarrow
	&& \fbox{$r < \frac{\delta}{a} \cdot \frac{(2 / \nu + a)(1 - \g^{-i})}{\beta \cdot (2/(a\nu) + 1) + 2 / \nu - 1} - 1$}
	\end{array}
\]

Therefore, the overall preprocessing requirements of $\alg{HORN}$ are
\(
	S_{\alg{HORN}} ~,~ P_{\alg{HORN}} \in (k+1)\cdot n^{2-\beta}
	= n^{2-\beta+\frac{\log(k+1)}{\log(n)}}
	= n^{2-\beta+\order{1}}
\)
for $k\in\polylog(n)$.
\end{proof}



\begin{thebibliography}{10}


\bibitem{2013-Agarwal-Godfrey}
R. Agarwal, P. Godfrey.
\newblock Distance oracles for stretch less than 2.
\newblock {\em SODA 2013}, pp.~526-538.

\bibitem{2011-Bartal-Gottlieb-Kopelowitz-Lewenstein-Roditty}
Y. Bartal, L.A. Gottlieb, T. Kopelowitz, M. Lewenstein, L. Roditty.
\newblock Fast, precise and dynamic distance queries.
\newblock {\em SODA 2011}, pp.~840-853.

\bibitem{2014-Bast-Delling-Goldberg-Hannemann-Pajor-Sanders-Wagner-Werneck}
H. Bast, D. Delling, A.~V. Goldberg, M. M\"{u}ller-Hannemann, T. Pajor, P. Sanders, D. Wagner, R. Werneck.
\newblock Route planning in transportation networks.
\newblock Technical Report MSR-TR-2014-4, Microsoft Research, April 2015 (\url{http://arxiv.org/abs/1504.05140}).

\bibitem{2013-Batz-Geisberger-Sanders-Vetter}
G.~V. Batz, R. Geisberger, P. Sanders, C. Vetter.
\newblock Minimum time-dependent travel times with contraction hierarchies.
\newblock {\em ACM J.~of Exp.~Algorithmics}, \textbf{18}, 2013.

\bibitem{1966-Cooke-Halsey}
K.~Cooke, E.~Halsey.
\newblock The shortest route through a network with time-dependent intermodal transit times.
\newblock {\em Math.~Anal.~and Appl.}, \textbf{14}(3):493-498, 1966.

\bibitem{2004-Dean-b}
B.~C.~Dean.
\newblock Algorithms for minimum-cost paths in time-dependent networks with waiting policies.
\newblock {\em Networks}, \textbf{44}(1):41--46, 2004.

\bibitem{2004-Dean-a}
B.~C. Dean.
\newblock Shortest paths in FIFO time-dependent networks: Theory and algorithms.
\newblock Technical Report, MIT, 2004.

\bibitem{2012-Dehne-Omran-Sack}
F. Dehne, O.~T. Masoud, and J.~R. Sack.
\newblock Shortest paths in time-dependent {FIFO} networks.
\newblock {\em Algorithmica}, \textbf{62}(1-2):416--435, 2012.

\bibitem{2011-Delling_TDSHARC}
D. Delling.
\newblock {Time-Dependent SHARC-Routing}.
\newblock {\em Algorithmica}, \textbf{60}(1):60--94, 2011.

\bibitem{2009-Delling-Wagner}
D. Delling, D. Wagner.
\newblock Time-dependent route planning.
\newblock {\em Robust and Online Large-Scale Optimization},
\textbf{LNCS~5868}, pp.~207-230. Springer, 2009.

\bibitem{1969-Dreyfus}
S.~E. Dreyfus.
\newblock An appraisal of some shortest-path algorithms.
\newblock {\em Operations~Research}, \textbf{17}(3):395--412, 1969.

\bibitem{2014-Foschini-Hershberger-Suri}
L. Foschini, J. Hershberger, S. Suri.
\newblock On the complexity of time-dependent shortest paths.
\newblock {\em Algorithmica}, \textbf{68}(4):1075--1097, 2014.

\bibitem{1977-Halpern}
J.~Halpern.
\newblock Shortest route with tme dependent length of edges and limited delay possibilities in nodes.
\newblock {\em Zeitschrifl f\"{u}r Operations Research}, \textbf{21}:117--124, 1977.

\bibitem{2015-Kontogiannis-Michalopoulos-Papastavrou-Paraskevopoulos-Wagner-Zaroliagis}
S. Kontogiannis, G. Michalopoulos, G. Papastavrou, A. Paraskevopoulos, D. Wagner, C. Zaroliagis.
\newblock Analysis and experimental evaluation of time-dependent distance oracles.
\newblock {\em ALENEX 2015}, pp.~147-158.

\bibitem{2016-Kontogiannis-Michalopoulos-Papastavrou-Paraskevopoulos-Wagner-Zaroliagis}
S. Kontogiannis, G. Michalopoulos, G. Papastavrou, A. Paraskevopoulos, D. Wagner, C. Zaroliagis.
\newblock Engineering oracles for time-dependent road networks.
\newblock {\em ALENEX 2016}, pp.~1-14.

\bibitem{2014-Kontogiannis-Zaroliagis}
S. Kontogiannis, C. Zaroliagis.
\newblock Distance oracles for time-dependent networks.
\newblock \emph{Algorithmica}, Vol.~\textbf{74}~(2016), No.~4, pp.~1404-1434.

\bibitem{ndls-bastd-12}
G. Nannicini, D. Delling, L. Liberti, D. Schultes.
\newblock Bidirectional {A*} search on time-dependent road networks.
\newblock {\em Networks}, \textbf{59}:240--251, 2012.

\bibitem{2014-Omran-Sack}
M. Omran, J.~R. Sack.
\newblock Improved approximation for time-dependent shortest paths.
\newblock {\em COCOON 2014}, \textbf{LNCS~8591}, Springer, pp.~453-464.

\bibitem{1990-Orda-Rom}
A. Orda, R. Rom.
\newblock Shortest-path and minimum delay algorithms in networks with time-dependent edge-length.
\newblock {\em J. of the ACM}, \textbf{37}(3):607--625, 1990.

\bibitem{2010-Patrascu-Roditty}
M. Patrascu, L. Roditty.
\newblock Distance oracles beyond the {T}horup--{Z}wick bound.
\newblock {\em FOCS 2010}, pp.~815-823.

\bibitem{2011-Porat-Roditty}
E. Porat, L. Roditty.
\newblock Preprocess, set, query!
\newblock {\em ESA 2011}, \textbf{LNCS~6942}, Springer, pp.~603-614.

\bibitem{SOS1998}
H.~Sherali, K.~Ozbay, S.~Subramanian.
\newblock The time-dependent shortest pair of disjoint paths problem: Complexity, models, and algorithms.
\newblock {\em Networks}, \textbf{31}(4):259--272, 1998.

\bibitem{2014-Sommer-spq-survey}
C. Sommer.
\newblock Shortest-path queries in static networks.
\newblock {\em ACM Comp.~Surv.}, \textbf{46}, 2014.

\bibitem{2009-Sommer-Verbin-Yu}
C. Sommer, E. Verbin, W. Yu.
\newblock Distance oracles for sparse graphs.
\newblock {\em FOCS 2009}, pp.~703-712.

\bibitem{2004-Thorup}
M. Thorup.
\newblock Compact oracles for reachability and approximate distances in planar digraphs.
\newblock {\em J. of the ACM}, \textbf{51}(6):993--1024, 2004.

\bibitem{2005-Thorup-Zwick}
M. Thorup, U. Zwick.
\newblock Approximate distance oracles.
\newblock {\em J. of the ACM}, \textbf{52}(1):1--24, 2005.

\bibitem{2012-Wulf-Nilsen-a}
C. Wulff-Nilsen.
\newblock Approximate distance oracles with improved preprocessing time.
\newblock {\em SODA 2012}, pp.~202-208.

\bibitem{2012-Wulf-Nilsen-b}
C. Wulff-Nilsen.
\newblock Approximate distance oracles with improved query time.
\newblock \emph{SODA 2013}, pp.~539--549.


\end{thebibliography}
\end{document}